\documentclass[a4paper,10pt]{article}

\usepackage{geometry}
\newgeometry{left=1in, right=1in, top=1in, bottom=1in}

\usepackage[T1]{fontenc}
\usepackage[utf8]{inputenc}
\usepackage{authblk}
\usepackage{commath}
\usepackage{braket}
\usepackage{multicol}
\usepackage{enumerate}
\usepackage{mathrsfs}
\usepackage{amsmath,amsthm,amssymb,xspace}
\usepackage{cite}
\usepackage{txfonts}
\usepackage{graphics}
\usepackage{float}
\usepackage{epstopdf}
\usepackage{epsfig}
\usepackage{caption}
\usepackage{cancel}
\usepackage[labelformat=simple]{subcaption}

\usepackage[colorlinks,bookmarksopen,bookmarksnumbered,citecolor=blue, linkcolor=black, urlcolor=black]{hyperref}

\usepackage{url}
\usepackage{pgfplots}
\usepgfplotslibrary{polar}
\usepgflibrary{shapes.geometric}
\usetikzlibrary{calc}

\usepackage{commath}
\usepackage{tikz}
\usepackage{tikz-3dplot}
\usepackage[qm]{qcircuit}
\usepackage{makecell}
\usepackage{pgf}

\usepackage{xcolor}
\usepackage[ruled,vlined,linesnumbered]{algorithm2e}
\usepackage{algorithmic}
\usepackage{multirow}
\newcommand{\poly}{\text{poly}}

\newcommand{\B}{\mbox{$\{0,1\}$}}
\newcommand{\Bn}{\mbox{$\{0,1\}^n$}}

\newcommand\dnk{\mbox{$\ket{D^n_k}$}\xspace}

\newcommand\Path{{\tt Path}}
\newcommand\Grid{{\tt Grid}}

\newcommand\D{{\sf Divide}}
\newcommand\U{{\sf U}}

\newtheorem{theorem}{Theorem}
\newtheorem{corollary}[theorem]{Corollary}
\newtheorem{lemma}[theorem]{Lemma}

\newtheorem{definition}[theorem]{Definition}

\title{Depth-Efficient Quantum Circuit Synthesis for Deterministic Dicke State Preparation}
 \author[]{Pei Yuan \footnote{Email: peiyuan@tencent.com} }
\author[]{Shengyu Zhang \footnote{Email: shengyzhang@tencent.com}}
 \affil[]{Tencent Quantum Laboratory}
\date{}

\begin{document}

\maketitle
\begin{abstract}
    The $n$-qubit $k$-weight Dicke states \dnk, defined as the uniform superposition of all computational basis states with exactly $k$ qubits in state $\ket{1}$, form a basis of the symmetric subspace and represent an important class of entangled quantum states with broad applications in quantum computing. We propose deterministic quantum circuits for Dicke state preparation under two commonly seen qubit connectivity constraints: 
    \begin{enumerate}
        \item All-to-all qubit connectivity: our circuit has depth $O(\log(k)\log(n/k)+k)$, which improves the previous best bound of $O(k\log(n/k))$.
        \item Grid qubit connectivity ($(n_1\times n_2)$-grid, $n_1\le n_2$):
        \begin{enumerate}
            \item For $k\ge n_2/n_1$, we design a circuit with depth $O(k\log(n/k)+n_2)$, surpassing the prior $O(\sqrt{nk})$ bound.
            \item For $k< n_2/n_1$, we design an optimal-depth circuit with depth $O(n_2)$. 
        \end{enumerate}
    \end{enumerate} 
    Furthermore, we establish the depth lower bounds of $\Omega(\log(n))$ for all-to-all qubit connectivity and $\Omega(n_2)$ for $(n_1\times n_2)$-grid connectivity constraints, demonstrating the near-optimality of our constructions.
\end{abstract}

\section{Introduction}
Quantum algorithms harness fundamental phenomena such as entanglement and coherence to achieve computational advantages over their classical counterparts. Over the past decades, Over the past decades, significant progress has been made in developing quantum algorithms for machine learning \cite{biamonte2017quantum}, solving linear and differential equations \cite{harrow2009quantum,berry2014high,childs2020quantum,an2023linear} and simulating Hamiltonians \cite{berry2015simulating,low2017optimal,low2019hamiltonian}. A critical component in many of these algorithms is quantum state preparation, which encodes a $2^n$-dimensional complex vector into an $n$-qubit quantum state. General quantum state preparation has been extensively investigated with optimal bounds established \cite{zhang2022quantum,sun2023asymptotically,yuan2023optimal,yuan2024does,luo2024circuit,zi2025constant} in recent years. %\pei{(pei:add some explanations of QSP.)}

While general quantum state preparation requires circuits of exponential depth or ancilla qubits—rendering even optimal constructions impractical—many quantum algorithms rely on specific entangled states that admit significantly more efficient implementations. A prominent example is the Dicke state, the uniform superposition of all computational basis states with a fixed Hamming weight. Formally, the Dicke state preparation problem is defined as follows: For any integer $0\le k \le n$, given an $n$-qubit initial state $\ket{0}^{\otimes n}$, prepare the $(n,k)$-Dicke state
\begin{equation}\label{eq:Dicke_state}
        \ket{D^n_k}:=\frac{1}{\sqrt{\binom{n}{k}}}\sum_{x\in \{0,1\}^n: \atop |x|=k}\ket{x},
    \end{equation}
where $|x|$ denotes the Hamming weight of the $n$-bit string $x$, i.e. the number of $1$'s. Notably, the $(n,n-k)$-Dicke state can be easily obtained by applying $X^{\otimes n}$ to the $(n,k)$-Dicke state. Thus, without loss of generality, we restrict our analysis to $0\le k\le \lfloor n/2\rfloor$ throughout this work. 

%\pei{(pei: the background and the importance of Dicke state preparation.)}  
Dicke states play a vital role across diverse domains of quantum information science. They are fundamental to quantum networks, quantum tomography, and quantum game theory \cite{dicke1954coherence,murao1999quantum,childs2000finding,ozdemir2007necessary,toth2010permutationally}. In quantum algorithms, Dicke states serve as initial states in variational quantum algorithms such as the Quantum Alternating Operator Ansatz (QAOA) to solve the $k$-vertex cover problem \cite{cook2020quantum}. They also enable key applications in quantum coding theory, including permutation-invariant quantum codes for quantum deletion channels \cite{ouyang2021permutation} and quantum error correction protocols \cite{ouyang2014permutation,ouyang2022finite}. 

Owing to their broad applicability, small-scale Dicke states have been experimentally demonstrated in various physical systems over the past two decades, including trapped ions \cite{hume2009preparation,lamata2013deterministic}, atomic ensembles \cite{xiao2007generation}, photonic systems \cite{wieczorek2009experimental} and superconducting circuits \cite{aktar2022divide}.
In addition to experimental results, theoretical results of circuit complexity for Dicke state preparation have also been extensively investigated \cite{bartschi2019deterministic,wang2021preparing,bartschi2022short,aktar2022divide,buhrman2024state,piroli2024approximating,yu2024efficient}. Quantum circuit costs are typically measured by size (gate count), depth (layer count), and the number of ancilla are the typical cost measures, corresponding to the preparation time complexity, execution time complexity and the space complexity of the circuit, respectively. For Dicke state, 
%
%The quantum circuit construction for Dicke state preparation has been explored in previous works.
Ref. \cite{cruz2019efficient}  presented a quantum circuit of depth $O(\log (n))$ and size $O(n)$ to prepare the W-state, the special Dicke state $\ket{D^n_1}$.
Under path graph connectivity, i.e. two-qubit gates can be applied only on qubits $i$ and $i+1$ for some $i\in\{1,\ldots,n-1\}$, \cite{bartschi2019deterministic} proposed circuits for Dicke state \dnk with $O(n)$ depth and  $O(nk)$ size.  
If all-to-all qubit connectivity is available, then the depth can be reduced to $O(k\log(n/k))$ \cite{bartschi2022short}. For 2D-grid connectivity graph, which is a commonly seen one for many physical implementation of quantum computers, \cite{bartschi2022short} showed a construction with depth $O(\sqrt{nk})$ when the grid is of size $n_1\times n_2$ with $n_2/n_1\le k \le n/2$ and $n_1n_2=n$. 

\paragraph{Beyond unitary circuits.} All the above Dicke state preparation circuits do not include measurement and ancillary qubits. Quantum circuit complexity was also studied when measurements are allowed or ancilla are available. Ref. \cite{buhrman2024state} gave a protocol to prepare the Dicke state in Local Alternating Quantum-Classical Computations (\textsf{LAQCC}), which consists of alternating layers of quantum and classical circuits and measurement, and is constrained to a grid connectivity constraint. 
For $k=O(\sqrt{n})$, the paper showed that an $(n,k)$-Dicke state can be prepared by a quantum circuit of depth $O(1)$ in \textsf{LAQCC} using $O(n^2\log(n))$ ancillary qubits, or $O(\log(n))$ for arbitrary $k$ with $O(\poly(n))$ ancillary qubits. Ref. \cite{yu2024efficient} further reduced depth to  $O(poly\log(n))$ with $O(poly\log(n))$ ancilla, later optimized to $O(n\log(n))$ ancilla \cite{liu2024low}. %\pei{Added: Recently Ref. \cite{zi2025constant} showed that any $n$-qubit quantum state can be prepared by a quantum circuit of depth $O(1)$ by using $O(2^n)$ anillary qubits and constant layers of measurements, which implies a constant-depth Dicke state preparation circuit.  }

\paragraph{Our results.} In this paper we focus on \textit{deterministic} quantum circuits for preparing $(n,k)$-Dicke states without measurements or ancilla, under two commonly seen qubit connectivity models: 
\begin{enumerate}
    \item All-to-all: Our circuit has a depth of $O(\log(k)\log(n/k)+k)$, which improves the previous best bound of $O(k\log(n/k))$ \cite{bartschi2022short}. 
    \item We then consider the $(n_1\times n_2)$-grid connectivity constraint, for which one can assume without loss of generality that $n_1\le n_2$. 
    \begin{enumerate}
        \item For $k\ge n_2/n_1$, we construct a circuit of depth $O(k\log(n/k)+n_2)$, which surpassing the previous one of $O(\sqrt{nk})$ \cite{bartschi2022short}. 
        \item For $k<n_2/n_1$, we design a circuit of depth $O(n_2)$, which is provably optimal. 
    \end{enumerate}
\end{enumerate} 
The optimality comes from our lower bound results. Specifically, we prove the depth lower bounds $\Omega(\log(n))$ and $\Omega(n_2)$ for all-to-all and $(n_1\times n_2)$-grid qubit connectivity constraint, respectively. {We conjecture that $\Omega(k)$ is also a lower bound even for all-to-all qubit connectivity, which would imply that our constructions are all depth-optimal (up to a logarithmic factor).} The results of circuit depth for the Dicke state preparation are summarized in Table \ref{tab:Dicke}. 

Our circuit results also extend to generation of arbitrary symmetric states composed of computational basis of Hamming weight at most $k$, while preserving depth cost for both connectivity models. 

%\syz{Let's summarize the previous and our results in a table for easy comparison.}

\begin{table}[]
    \centering
    
    \begin{tabular}{c c c c }
    \hline
    \hline
        qubit connectivity & r  esults & range of $k$ & circuit depth    \\
    \hline
        \multirow{4}{*}{all-to-all} &\cite{cruz2019efficient} & $1$ & $O(\log(n))$    \\
        %\cite{bartschi2019deterministic} & $1\le k \le n/2$ & $O(n)$ &  $n$-path \\
        &\cite{bartschi2022short} & $[1,n/2]$ & $O(k\log(n/k))$  \\
        &\textbf{ours} (Thm. \ref{thm:Dicke_unitary_noconstraint_improve}) & $[1,n/2]$ & $O(\log(k)\log(n/k)+k)$    \\
        &\textbf{ours} (Thm.  \ref{thm:lb_diam}) & $[1,n/2]$ & $\Omega(\log(n))$   \\

        %\multirow{3}{*}{\cite{buhrman2024state}} & $k=1$ & $O(1)$ & grid & Non-deterministic& $O(n\log(n))$\\
         %& $k=O(\sqrt{n})$ & $O(1)$ &  grid & Non-deterministic & $O(n^2\log(n))$\\
         %& $1\le k \le n/2$ & $O(\log(n))$ &  grid & Non-deterministic & $O(\poly(n))$\\
        %\cite{yu2024efficient} & $1\le k \le n/2$ &  $\poly(\log(n))$ & all-to-all & Non-deterministic &$\poly(\log(n))$\\
        %\cite{liu2024low} & $1\le k \le n/2$ & $O(\log(n))$ & all-to-all & Non-deterministic & $O(n\log(n))$\\
        \hline
         \multirow{4}{*}{$ (n_1\times n_2)$-grid} &\cite{bartschi2022short} & $[n_2/n_1, n/2]$ & $O(\sqrt{nk})$  \\
        &\textbf{ours} (Coro.\ref{coro:Dicke_grid_improve})  & $[n_2/n_1, n/2]$ & $O(k\log(n/k)+n_2)$    \\
        &\textbf{ours} (Coro.\ref{coro:Dicke_grid_improve}) & $[1, n_2/n_1]$ & $O(n_2)$    \\
        
        &\textbf{ours} (Thm.\ref{thm:lb_diam}) & $[1,n/2]$ & $\Omega(n_2)$    \\
        \hline
        \hline
    \end{tabular}
    \caption{The circuit depth of deterministic preparation of $(n,k)$-Dicke state ($1\le k\le n/2$) without ancilla or measurement for the complete and 2D-grid graphs. The 2D-grid is of dimension $n_1\times n_2$, with $n_1\le n_2$ and $n=n_1n_2$.}
    \label{tab:Dicke}
\end{table}

%\syz{Can we generalize the result to generating an arbitrary state $\ket{\psi}$ in the symmetric subspace? Write $\ket{\psi}$  as a linear combination of Dicke's state $\sum_{k=0}^n \alpha_k \ket{D_k^n}$ and use LCU. One possible way is first to generate the coefficient state $\sum_{k=0}^n \alpha_k \ket{k}$ by QSP, and then use CQSP to get $\sum_{k=0}^n \alpha_k \ket{k}\ket{D_k^n}$, undo the QSP to get $\ket{\psi}$ as a component, and finally use amplitude amplification.}

The remainder of this paper is structured as follows. Section \ref{sec:pre} introduces key notations and reviews relevant prior work. In Section \ref{sec:circuit}, we present our main results: quantum circuits for Dicke state preparation both with all-to-all connectivity and under grid connectivity constraints. Section \ref{sec:lowerbound} establishes fundamental depth lower bounds for these preparation schemes. We conclude with a summary of our findings and discuss potential extensions in Section \ref{sec:conclusion}.

\section{Preliminaries}
\label{sec:pre}
This section introduces key notations and relevant results used throughout the paper.

\paragraph{Notation} Let $\Bn$ denote the set of all $n$-bit strings. We define $[n] = \{1,2,\ldots, n\}$ and $[n]_0 = \{0,1,2,\ldots, n\}$. For a bit string $x\in \Bn$, its Hamming weight $|x|$ counts the number of 1s in $x$. For a qubit set of qubits $S\subseteq [n]$, denote by $\ket{\psi}_S$ an $|S|$-qubit state $\ket{\psi}$ supported on qubits in $S$. When $S = \{i\}$ is a singleton, we simplify this to $\ket{\psi}_{i}$. 

We define the following quantum gates. 
\begin{enumerate}
    \item Toffoli gate ${\sf Tof}^{S}_{t}(x)$: an $(|S|+1)$-qubit gate where $S$ is the set of control qubits, $t$ is the target qubit and $x\in\{0,1\}^{|S|}$ gives the activation pattern. The gate is defined as ${\sf Tof}^{S}_{t}(x)\ket{y}_S\ket{a}_t=\ket{y}_S\ket{a\oplus [x=y]}_t$ with $[x=y]$ is the indicator function ($[x=y]=1$ if $x=y$ and $[x=y]=0$ otherwise). 
    \item CNOT gate ${\sf CNOT}^{s}_{t}$: ${\sf CNOT}^s_t\ket{x}_s\ket{y}_t=\ket{x}_s\ket{x\oplus y}_t$.  
    \item SWAP gate ${\sf SWAP}^{s}_{t}$: ${\sf SWAP}^{s}_{t}\ket{x}_s\ket{y}_t=\ket{y}_s\ket{x}_t$. 
\end{enumerate}

Throughout this work, we consider \textit{standard quantum circuits} composed exclusively of 1- and 2-qubit gates. A circuit is called a \textit{CNOT circuit} if it contains only CNOT gates.

\paragraph{Qubit connectivity} We model qubit connectivity constraints using an undirected graph graph $G = (V,E)$, where the vertex set $V$ represents the set of qubits and the edge $E$ specifies allowed two-qubit interactions. A two-qubit gate can be applied to qubits $i,j\in V$ if and only if $(i,j)\in E$. We refer to $G$ as the \textit{constraint graph} of the circuit and we say that the circuit is \textit{under $G$ constraint}. Important special cases include the following. 
\begin{enumerate}
    \item All-to-all qubit connectivity: $G = K_n$, the complete graph.
    \item 2D-grid connectivity: $G = \Grid_n^{n_1,n_2}$, an $(n_1\times n_2)$-grid with $n=n_1n_2$ qubits (assuming $n_1\le n_2$ without loss of generality). 
    \item Linear connectivity: $G = \Path_n$, an $n$-vertex path graph.
\end{enumerate}

We summarize several known circuit implementations that will be used in our constructions.
\begin{lemma} [\cite{baker2019decomposing,multi-controlled-gate}]\label{lem:tof}
An $n$-qubit Toffoli gate admits two implementations: (1) it can be implemented by a standard quantum circuit of $O(n)$ depth and size without using any ancillary qubits, and (2) also by one with $O(\log n)$ depth and $O(n)$ size using $n-1$ ancillary qubits.
\end{lemma}

\begin{lemma}[\cite{sun2023asymptotically}]\label{lem:CNOT_sum}
A unitary transformation $U_{add}(S,t)$ implementing
\begin{equation}\label{eq:CNOT_sum}
    \ket{x_1x_2\cdots x_n}_S\ket{k}_t\xrightarrow{U_{add}(S,t)} \ket{x_1x_2\cdots x_n}_S|\oplus_{i=1}^n x_i \oplus k\rangle_t, \qquad \forall x_1,\ldots,x_n,~k\in\B
\end{equation}
can be realized by a standard quantum circuit of depth $O(\log(n))$.
\end{lemma}

\begin{lemma}[\cite{sun2023asymptotically}]\label{lem:copy}
A copying unitary $U_{copy}$ satisfying
\begin{equation}
  \ket{x}\ket{0^{tn}}\xrightarrow{U_{copy}} \ket{x}\underbrace{\ket{x}\ket{x}\cdots \ket{x}}_{t~\text{copies of}~\ket{x}},\qquad x\in \Bn  
\end{equation}
admits a CNOT circuit of depth $O(\log t)$ and size $O(tn)$.
\end{lemma}

\begin{lemma}[\cite{jiang2020optimal}]\label{lem:cnotcircuit}
    Any $n$-qubit CNOT circuit can be parallelized to depth $O\left(\log(n)+\frac{n^2}{(n+m)\log(n+m)}\right)$ using $m\ge 0$ ancillary qubits.
\end{lemma}

\begin{lemma}[\cite{yuan2024full}]\label{lem:permutation}
    For any permutation $\pi\in S_n$, the corresponding \textit{permutation unitary} $U_\pi^n$, defined as
    \begin{equation}\label{eq:permutation}
        U^n_\pi\ket{x_1x_2\cdots x_n}=\ket{x_{\pi(1)}x_{\pi(2)}\cdots x_{\pi(n)}}, \qquad \forall x_i\in\{0,1\}^n, \quad\forall i\in[n],
    \end{equation}
    can be implemented by a standard quantum circuit consisting of depth $O(n_2)$ under $\Grid_n^{n_1,n_2}$ constraint. 
\end{lemma}
\begin{lemma}
[\cite{yuan2023optimal}] \label{lem:multi-QSP}
For any integers $k,m\ge 0$, $n>0$ and any $n$-qubit quantum states $\{\ket{\psi_x}:x\in \{0, 1\}^k\}$, the following $(k,n)$-controlled quantum state preparation, or $(k,n)$-CQSP, 
\begin{equation}    
    \ket{x}\ket{0^{n}}\to\ket{x}\ket{\psi_x}, \qquad \forall x\in \{0, 1\}^k
\end{equation}
can be implemented by a standard quantum circuit of depth $O\left(n+k+\frac{2^{n+k}}{n+k+m}\right)$ with $m$ ancillary qubits. 
\end{lemma}

\section{Quantum circuit for Dicke state preparation}
\label{sec:circuit}

This section presents out circuit constructions for preparing Dicke states. We begin by recalling a basic framework from prior work \cite{bartschi2019deterministic, bartschi2022short}, which our approach builds upon. Subsequent subsections detail optimized implementations for all-to-all qubit connectivity (Section \ref{sec:all-to-all}) and grid constrained connectivity (Section \ref{sec:grid}).

We first recall a unitary from \cite{bartschi2019deterministic}, %The following two lemmas show the circuit depth of two subcircuits, the Dicke state unitary and the divide operator, under $\Path_n$ connectivity constraints. These subcircuits are utilized in the circuit framework of Dicke state preparation. The first one is an 
\textit{$(n,k)$-Dicke state unitary} $\U^n_k(S)$, which acts on a qubit set $S$ of size $n$ and generates the $(n,\ell)$-Dicke state on input $\ket{0^{n-\ell} 1^\ell}$, for any $\ell \le k$. That is,
    \begin{align}\label{eq:Dicke_unitary}
        \U^n_k(S)\ket{0^{n-\ell} 1^\ell}_S=\ket{D^n_\ell}_S, \quad & \forall \ell \in [k]_0,
    \end{align}
where $\ket{D^n_\ell}$ is the $(n,\ell)$-Dicke state. {Note that this constitutes a slightly stronger requirement than the standard Dicke state \dnk preparation for a fixed $k$, as it needs to handle all $\ell\le k$ simultaneously.}
\begin{lemma}[\cite{bartschi2019deterministic}]\label{lem:Dicke_unitary_path}
    The $(n,k)$-Dicke state unitary $\U^n_k(S)$ can be implemented by a standard quantum circuit of depth $O(n)$ and size $O(nk)$ under the $\Path_n$ constraint, without ancillary qubits.
\end{lemma}

One crucial subroutine for preparing Dicke states is a unitary which creates a superposition of states $\ket{0^{k-i}1^i}\ket{0^{k+i-\ell}1^{\ell-i}}$ with different $i\le \ell$. More precisely, let $m\ge k$ and $n-m\ge k$, the \textit{divide unitary} $\D^{n,m}_k(S_1,S_2)$ operates on disjoint $k$-qubit sets $S_1$ and $S_2$, satisfying
\begin{align}\label{eq:wdb}
    \D^{n,m}_k(S_1,S_2)\ket{0^k}_{S_1}\ket{0^{k-\ell}1^{\ell}}_{S_2}=\frac{1}{\sqrt{\binom{n}{\ell}}}\sum_{i=0}^{\ell} \sqrt{\binom{m}{i}\binom{n-m}{\ell-i}}\ket{0^{k-i}1^i}_{S_1}\ket{0^{k+i-\ell}1^{\ell-i}}_{S_2}, \quad &  \forall \ell\in[k]_0,
\end{align}
with the convention $\binom{s}{t}=0$ if $s<t$.
\begin{lemma}[\cite{bartschi2022short}]\label{lem:WDB_path}
    The \textit{divide unitary} $\D^{n,m}_k(S_1,S_2)$ acting on $2k$ adjacent qubits can be implemented by a quantum circuit of depth $O(k)$ and size $O(k^2)$ under $\Path_{2k}$ constraint, using no ancillary qubits.
\end{lemma}
For notational convenience, we may drop the sets and shorten $\U^n_k(S)$ and $\D^{n,m}_k(S_1,S_2)$ to $\U^n_k$ and $\D^{n,m}_k$, respectively, when the sets are clear from the context.
% Note that when we apply $WDB^{n,m}_k$ on basis states $\ket{0^{n-\ell}1^\ell}$ for $\ell\in[k]_0$, at most $2k$ qubits are unchanged. If $m\ge k$ and $n-m\ge k$, then we can rewrite $WDB^{n,m}_k$ as $WDB^{n,m}_k(S_1,S_2)$ satisfying
% \begin{align}\label{eq:wdb_set}
%      WDB^{n,m}_k(S_1,S_2)\ket{0^k}_{S_2}\ket{0^{k-\ell} 1^\ell}_{S_1}=\sqrt{\binom{n}{\ell}}^{-\frac{1}{2}}\sum_{i=0}^{\ell} \sqrt{\binom{m}{i}\binom{n-m}{\ell-i}}\ket{0^{k-i}1^i}_{S_2}\ket{0^{k+i-\ell}1^{\ell-i}}_{S_1}, \quad &  \forall \ell\in[k]_0,
% \end{align}
% where $S_1$ and $S_2$ are qubit sets of sizes $k$. Namely, $WDB^{n,m}_k(S_1,S_2)$ only act on qubit sets $S_1$ and $S_2$.

% \begin{lemma}[\cite{bartschi2022short}]\label{lem:Dicke_unitary_noconstraint}
%     The Dicke state unitary $U^n_k$ can be implemented by a quantum circuit of depth $O(k\log(n/k))$ without qubit connectivity constraints.
% \end{lemma}

With the above setup, we now sketch the circuit framework of the $n$-qubit Dicke state unitary $\U^{n}_k(S)$ (in Lemma \ref{lem:Dicke_unitary_path}), which also underlies both our unconstrained (Theorem \ref{thm:Dicke_unitary_noconstraint_improve}) and grid-constrained (Theorem \ref{thm:Dicke_unitary}) optimizations. 
A Dicke state unitary $\U^n_k$ can be realized as follows. For any $\ell\in[k]_0$,
\begin{align*}
       \ket{0^{n-\ell}1^\ell}&=\ket{0^{\lfloor n/2\rfloor-k}}_{T_1}\ket{0^k}_{S_1}\ket{0^{\lceil n/2\rceil-k}}_{T_2}\ket{0^{k-\ell}1^\ell}_{S_2}\\
   \xrightarrow{\D^{n,\lfloor n/2\rfloor}_k(S_1,S_2)}&~ \frac{1}{\sqrt{\binom{n}{\ell}}}\sum_{i=0}^{\ell} \sqrt{\binom{\lfloor n/2 \rfloor}{i}\binom{\lceil n/2 \rceil}{\ell-i}}\ket{0^{\lfloor n/2\rfloor -k}}_{T_1}\ket{0^{k-i}1^i}_{S_1}\ket{0^{\lceil n/2\rceil-k}}_{T_2} \ket{0^{k+i-\ell}1^{\ell-i}}_{S_2}& \text{(by Eq. \eqref{eq:wdb})}\\
   \xrightarrow{\U^{\lfloor n/2 \rfloor}_{k}(T_1\cup S_1)\otimes \U^{\lceil n/2\rceil}_{k}(T_2\cup S_2)} &~  \frac{1}{\sqrt{\binom{n}{\ell}}}\sum_{i=0}^{\ell} \sqrt{\binom{\lfloor n/2\rfloor}{i}\binom{\lceil n/2\rceil}{\ell-i}}\ket{D^{\lfloor n/2\rfloor}_i}_{T_1\cup S_1}\ket{D^{\lceil n/2\rceil}_{\ell-i}}_{T_2\cup S_2} & \text{(by Eq. \eqref{eq:Dicke_unitary})}\\
   =&~ \frac{1}{\sqrt{\binom{n}{\ell}}}\sum_{i=0}^{\ell}\sum_{x_1: x_1\in \{0,1\}^{\lfloor n/2\rfloor}, \atop |x_1|=i}\ket{x_1}_{T_1\cup S_1}\sum_{x_2: x_2\in \{0,1\}^{\lceil n/2\rceil},\atop |x_2|=\ell-i}\ket{x_2}_{T_2\cup S_2} & \text{(by Eq. \eqref{eq:Dicke_state})}\\
   =&~ \frac{1}{\sqrt{\binom{n}{\ell}}} \sum_{x: x\in \{0,1\}^n,\atop |x|=\ell}\ket{x}=\ket{D^n_\ell} & \text{(by Eq. \eqref{eq:Dicke_state})}\\
   =&~ \U^n_k\ket{0^{n-\ell}1^\ell}. &\text{(by Eq. \eqref{eq:Dicke_unitary})}
\end{align*}
The above shows that a Dicke state unitary $\U^n_k$ admits a divide-and-conquer approach via a recursive decomposition into one divide unitary $\D^{n,\lfloor n/2\rfloor}_k$ and two smaller-scale Dicke state unitaries $\U^{\lfloor n/2\rfloor}_k$ and $\U^{\lceil n/2\rceil}_k$, i.e., 
\begin{align}\label{eq:framework}
    \U^n_k=(\U^{\lfloor n/2\rfloor}_k\otimes \U^{\lceil n/2\rceil}_k)\D^{n,\lfloor n/2\rfloor}_k    
\end{align} 
We can recurse on $\U^{\lfloor n/2\rfloor}_k$ and $\U^{\lceil n/2\rceil}_k$ until all the Dicke state unitaries consist of $O(k)$ qubits. Namely, $\U^n_k$ can be implemented by at most $\lfloor \log(n/k)\rfloor$ layers of divide unitaries and one layer of $O(k)$-qubit Dicke state unitaries, where the $j$-th layer consists of $2^{j-1}$ divide unitaries. Also see an example of circuit framework for $\U^9_2$ in Fig. \ref{fig:unk_framework}.

\begin{figure}[hbt]
\centerline 
{\Qcircuit @C=1em @R=0.8em {
&  \multigate{8}{\U^9_2}& \qw\\
&  \ghost{\U^9_2}& \qw\\
&  \ghost{\U^9_2}& \qw\\
&  \ghost{\U^9_2}& \qw\\
&  \ghost{\U^9_2}& \qw & & & \Longrightarrow\\
&  \ghost{\U^9_2}& \qw\\
&  \ghost{\U^9_2}& \qw\\
&  \ghost{\U^9_2}& \qw\\
&  \ghost{\U^9_2}& \qw\\
}
\\
\Qcircuit @C=1em @R=0.8em {
& \qw  & \multigate{3}{\U^4_2}& \qw\\
& \qw & \ghost{\U^4_2}& \qw\\
& \multigate{1}{\D^{9,4}_2}  & \ghost{\U^4_2}& \qw\\
& \ghost{\D^{9,4}_2}  &  \ghost{\U^4_2}& \qw\\
 & \qw\qwx[3]\qwx[-1]   & \multigate{4}{\U^5_2}& \qw & & & \Longrightarrow\\
 & \qw  & \ghost{\U^5_2}& \qw\\
 & \qw  & \ghost{\U^5_2}& \qw\\
 & \multigate{1}{\D^{9,4}_2}  &\ghost{\U^5_2}& \qw\\ 
 & \ghost{\D^{9,4}_2}  & \ghost{\U^5_2} & \qw\\
}
\\
\Qcircuit @C=1em @R=0.8em {
& \qw &  \multigate{1}{\D^{4,2}_2} & \multigate{1}{\U^2_2}& \qw\\
& \qw & \ghost{\D^{4,2}_2} \qwx[1]& \ghost{\U^2_2}& \qw\\
& \multigate{1}{\D^{9,4}_2} & \multigate{1}{\D^{4,2}_2} & \multigate{1}{\U^2_2}& \qw\\
& \ghost{\D^{9,4}_2} & \ghost{\D^{4,2}_2} &  \ghost{\U^2_2}& \qw\\
 & \qw  & \multigate{1}{\D^{5,2}_2} & \multigate{1}{\U^2_2}& \qw\\
 & \qw & \ghost{\D^{5,2}_2} \qwx[2]& \ghost{\U^2_2}& \qw\\
 & \qw & \qw & \multigate{2}{\U^3_2}& \qw\\
 & \multigate{1}{\D^{9,4}_2} \qwx[-4] &  \multigate{1}{\D^{5,2}_2} &\ghost{\U^3_2}& \qw\\ 
 & \ghost{\D^{9,4}_2}  & \ghost{\D^{5,2}_2} & \ghost{\U^3_2} & \qw\\
}}
    \caption{An example of circuit framework for $\U^9_2$.}
    \label{fig:unk_framework}
\end{figure}
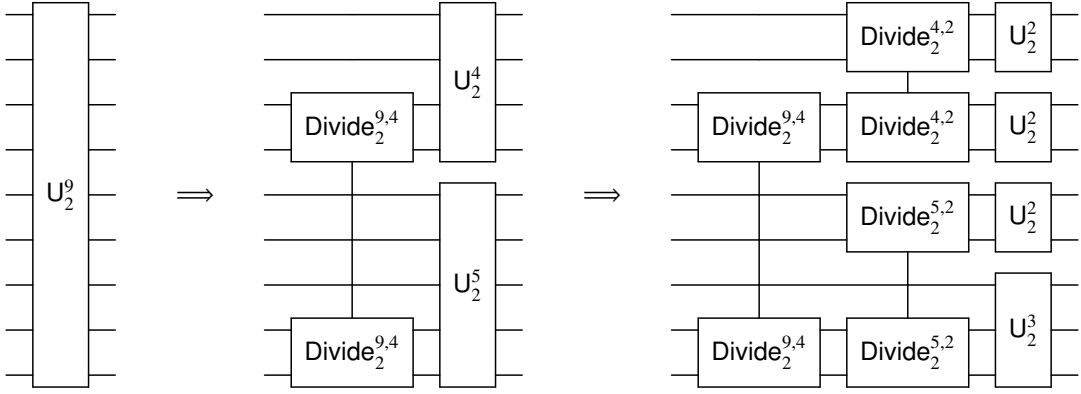

% \begin{lemma}
%     Let $x=x_1x_2\cdots x_n\in\Bn$ and $t\in\B$. Unitaries $U_{add}$ and $U_{copy }$ satisfy
%     \begin{align*}
%         & U_{add}\ket{x}\ket{t}=\ket{x}\ket{t\oplus \bigoplus_{j=1}^n x_i}, & \forall x\in \Bn, t\in \B,\\
%         & U_{copy}\ket{x}\ket{0^n}\cdots \ket{0^n}=\ket{x}\ket{x}\cdots\ket{x}, &  \forall x\in \Bn.
%     \end{align*}
% \end{lemma}

\subsection{Dicke state preparation with all-to-all qubit connectivity}
\label{sec:all-to-all}
The above framework \cite{bartschi2022short} for implementing the $\U^n_k$ consists of $O(\log(n/k))$ layers of $2k$-qubit divide unitaries followed by one layer of small-scale Dicke state unitaries $\U^{O(k)}_k$. This construction yields an overall circuit depth of $O(k \log (n/k))$. Our work achieves significant depth compression with three key ingredients: First, we observe that in each layer, the divide unitaries act only on a subset of qubits, allowing idle qubits to serve as temporary ancilla. %This insight enables us to compress the circuit depth of divide operators by utilizing these idle ancillary qubits.\\
Second, we adopt a hybrid encoding approach: we use one-hot encoding to move 1s for better parallelization, binary encoding for efficient superposition generation. The transform between these encodings as well as the unary encoding in the divide unitary can be implemented in low depth. 

Our improved divide unitary implementation proceeds through four phases: 
(1) Encode the input basis states in binary form (Lemma \ref{lem:binary}), and then use a CQSP circuit to create the superposition  share the same amplitude as the RHS of Eq. \eqref{eq:wdb}.
(2) Convert the binary encoding into a more sparse one-hot encoding, allowing more parallel implementation of the divide unitary's $1$-bit distribution in the basis states. 
(3) Use plus and minus operations (Lemmas \ref{lem:unitary_minus} and \ref{lem:unitary_plus}) on the one-hot encoded basis to effectively move 1s to the right positions.
(4) Transform one-hot to unary encoding as required by the divide unitary.
 
Next we will present the formal construction, starting at a few lemmas for achieving the above encoding transform and plus/minus operations.

%In this section, we improve the depth of the Dicke state preparation circuit without qubit connectivity. As mentioned above, since each divide operator exclusively acts on $2k$ qubits and there are $2^{j-1}$ divide operators in the $j$-th layer of the circuit framework, the total number of idle qubits in the $j$-th layer amounts to $n-k2^{j}$. For example, there are $9-2\cdot 2^1=5$ idle qubits in the first layer of the circuit framework for $\U^9_2$ in Figure \ref{fig:unk_framework}. Since $\ket{D^n_k}$ is prepared by applying $\U^n_k$ on quantum state $\ket{0^{n-k}1^k}$, all idle qubits are initialized in the $\ket{0}$ state. Consequently, the idle qubits can be leveraged as ancillary qubits in each layer of the circuit framework. Ref. \cite{bartschi2022short} constructed the circuits for the divide operator as Lemma \ref{lem:WDB_path}. In this section, we reduce the circuit depth of the divide operator by using the idle qubits.

%First, we propose the circuit constructions of three unitary transformations in the following Lemmas \ref{lem:binary}, \ref{lem:unitary_minus} and \ref{lem:unitary_plus}, which are utilized to construct the circuit of the divide operator.

For any number $\ell\in[k]_0$, there are three natural encodings of $\ell$ by strings in $\B^k$
\begin{enumerate}
    \item binary encoding: $\ket{0^{k-\lceil \log(k+1)\rceil}(\ell)_2}$, where $(\ell)_2$ denotes the binary representation of integer $\ell$ by $\ell_{\lceil \log(k+1)\rceil}\cdots \ell_2 \ell_1\in \B^{\lceil\log(k+1)\rceil}$ with $\ell=\sum_{j=1}^{\lceil \log(k+1)\rceil}\ell_j2^{j-1}$.
    \item one-hot encoding: $\ket{0^{k-\ell}10^{\ell-1}}$, i.e. the $\ell$-th position (from right) is $1$ and all others are $0$.
    \item unary encoding: $\ket{0^{k-\ell}1^\ell} $, i.e. the first $\ell$ positions (from right) is $1$ and all others are $0$.
\end{enumerate}

\begin{lemma}\label{lem:binary}
    With all-to-all qubit connectivity and $N\ge 2k$ ancillary qubits, the change of unary and one-hot encoding basis 
    \begin{equation}\label{eq:U_uo}
        \ket{0^{k-\ell}1^\ell} \to \ket{0^{k-\ell}10^{\ell-1}}, \quad {\forall \ell\in [k]_0,}
    \end{equation}
    can be realized by a standard quantum circuit $U_{uo}$ of depth $O\big(\log(k)+\frac{k^2}{(N+k)\log(N+k)}\big)$, and the change of one-hot and binary encoding basis 
    \begin{equation}\label{eq:U_ob}
      \ket{0^{k-\ell}10^{\ell-1}} \to \ket{0^{k-\lceil \log(k+1)\rceil}(\ell)_2}, \quad {\forall \ell\in [k]_0,}  
    \end{equation}
    can be realized by a standard quantum circuit $U_{ob}$ of depth $O\left(\log(k)+\frac{k^2}{N+k}\right)$.
    
    %There is a unitary transformation $U_{binary}$ satisfies
    %\[U_{binary}\ket{0^{k-\ell}1^\ell} = \ket{0^{k-\lceil \log(k+1)\rceil}(\ell)_2}, \quad \forall \ell\in[k]_0.\]
    %Then $U_{binary}$ can be realized by a quantum circuit of depth $O\left(\log(k)+\frac{k^2}{N+k}\right)$ using $N\ge 2k$ ancillary qubits without qubit connectivity constraints.
\end{lemma}
\begin{proof}
     The $k$ input qubits of $U_{uo}$ are labelled as qubit set $S=\{s_k,s_{k-1},\ldots,s_1\}$. Unitary $U_{uo}$ (Eq. \eqref{eq:U_uo}) can be realized by a CNOT circuit $\prod_{j=1}^{k-1}{\sf CNOT}^{s_{j+1}}_{s_j}$, whose circuit depth can be reduced to  $O\big(\log(k)+\frac{k^2}{(N+k)\log(N+k)}\big)$ using $N$ ancillary qubits according to Lemma \ref{lem:cnotcircuit}.

   To construct the circuit of unitary $U_{ob}$, the $k$ input and $N$ ancillary qubits are labelled as follows:
   The first $k$ input qubits are labelled as qubit set $S=\{s_k,s_{k-1},\ldots,s_1\}$. The first $k$ ancillary qubits are labelled as qubit set $T=\{t_k,t_{k-1},\ldots, t_1\}$. Let $p:=\lfloor\frac{N-k}{\lceil \log(k+1)\rceil}\rfloor$. The second $p\cdot\lceil\log(k+1)\rceil=O(N-k)$ ancillary qubits are divided into $p$ parts of size $\lceil \log(k+1)\rceil$, which are labelled as $A_1,A_2,\ldots, A_p$. The unitary $U_{ob}$ can be implemented in the following two steps:
   %The unitaries \pei{$U_{uo}$ and $U_{ob}$} can be implemented as follows by using ancillary qubits:
    \begin{align}
        \ket{0^{k-\ell}10^{\ell-1}}_S\ket{0^{k}}_T\bigotimes_{j=1}^p \ket{0^{\lceil\log(k+1)\rceil}}_{A_j}
        & \to \ket{0^{k-\lceil\log(k+1)\rceil}(\ell)_2}_S \ket{0^{k-\ell}10^{\ell-1}}_T\bigotimes_{j=1}^p \ket{0^{\lceil\log(k+1)\rceil}}_{A_j},\label{eq:1-3}\\
        & \to \ket{0^{k-\lceil\log(k+1)\rceil}(\ell)_2}_S \ket{0^k}_T\bigotimes_{j=1}^p \ket{0^{\lceil\log(k+1)\rceil}}_{A_j}, \label{eq:1-4}
    \end{align}
     for any $\ell\in[k]_0$. In above equations, define $\ket{0^{k-\ell}10^{\ell-1}}$ as $\ket{0^k}$ if $\ell=0$. To implement Eq. \eqref{eq:1-3}, first we apply $\prod_{i=1}^k \prod_{j:~i_j=1, \atop j\in[\lceil\log(k+1)\rceil]} {\sf CNOT}^{s_i}_{t_j} $, which transforms basis $\ket{0^{k-\ell}10^{\ell-1}}_S\ket{0^{k}}_T$ to $\ket{0^{k-\ell}10^{\ell-1}}_S\ket{0^{k-\lceil \log(k+1)\rceil}(\ell)_2}_T$ for any $\ell\in[k]_0$. The depth of this CNOT circuit can be reduced to $O\big(\log(k)+\frac{k^2}{N\log(N)}\big)$ using $O(N-k)$ ancillary qubits in qubit set $A_1\cup \cdots \cup A_p$ based on Lemma \ref{lem:cnotcircuit}. Second, we swap the state of $S$ and $T$ by a $1$-depth circuit $\prod_{j=1}^{k}{\sf SWAP}^{s_i}_{t_j}$, which completes the circuit construction of Eq. \eqref{eq:1-3}. Eq. \eqref{eq:1-4} can be implemented as follows. For any $\ell\in[k]_0$,
    \begin{align*}
         \ket{0^{k-\lceil \log(k+1)\rceil}(\ell)_2}_S\ket{0^{k-\ell}10^{\ell-1}}_{T} \bigotimes_{j=1}^p \ket{0^{\lceil\log(k+1)\rceil}}_{A_j}
       \to & \ket{0^{k-\lceil \log(k+1)\rceil}(\ell)_2}_S\ket{0^{k-\ell}10^{\ell-1}}_T \bigotimes_{j=1}^p \ket{(\ell)_2}_{A_j}\\
       \to & \ket{0^{k-\lceil \log(k+1)\rceil}(\ell)_2}_S\ket{0^k}_T\bigotimes_{j=1}^p \ket{(\ell)_2}_{A_j}\\
       \to & \ket{0^{k-\lceil \log(k+1)\rceil}(\ell)_2}_S \ket{0^k}_T\bigotimes_{j=1}^p \ket{0^{\lceil\log(k+1)\rceil}}_{A_j}
    \end{align*}
   %In the initial state, the register $A$ is composed of the first $\lceil\log(k+1)\rceil$ qubits of state. The register $T=\{t_1,t_2,\ldots,t_k\}$ consists of the second $k$ qubits. The $N$ ancillary qubits are divided into $N/\lceil\log(k+1)\rceil$ parts, and each part is called register $A_i$ for any $i\in N/\lceil\log(k+1)\rceil$. 
   The first line makes $p$ copies of $\ket{(\ell)_2}$ on the qubit sets $A_1,A_2,\ldots,A_p$, which can be implemented by a circuit of depth $O(\log p)$ based on Lemma \ref{lem:copy}. The second line is implemented as follows. We apply $\prod_{i=1}^{p}{\sf Tof}^{A_i}_{t_i}((i)_2)$ to transform the first $p$ qubits $\{t_p,t_{p-1},\ldots,t_1\}$ of register $T$ to $\ket{0^p}$. These Toffoli gates can be realized in parallel of depth $O(\log(k))$ based on Lemma \ref{lem:tof}, since they act on distinct qubits. 
   By similar discussion, each $p$ qubits in register $T$ can be transformed to $\ket{0^p}$ by a circuit of depth $O(\log(k))$. Therefore, the total depth for the second line is $O(\log(k))\cdot \lceil k/p\rceil=O(k\log(k)/p)$. The third line is the inverse of the first line, which has depth $O(\log(p))$. Then the total depth of Eq. \eqref{eq:1-4} is $2\cdot O(\log(p)+O(k\log(k)/p)=O(\log(N/\log(k))+k\log^2(k)/N)$. In summary, the total depth of unitary $U_{ob}$ is $O\big(\log(k)+\frac{k^2}{(N+k)\log(N+k)}\big)+O\big(\log(k)+\frac{k^2}{N\log(N)}\big)+O(\log(N/\log(k))+k\log^2(k)/N)=O\big(\log(k)+\frac{k^2}{N+k}\big)$.
   
   % In summary, the total depth of $U_{binary}$ is
   % \[O(1)+\big(\log(k)+\frac{k^2}{(N+k)\log(N+k)}\big)+ O\big(\log(k)+\frac{k^2}{N\log(N)}\big)+O\big(\log(N/\log(k))+\frac{k\log^2(k)}{N}\big) = O\big(\log(k)+\frac{k^2}{N+k}\big).\]

\end{proof}

\begin{lemma}\label{lem:unitary_minus}
A unitary transformation $U_{minus}$ satisfying
    \begin{equation}\label{eq:unitary_minus}
        U_{minus}\ket{0^{k-i}10^{i-1}}\ket{0^{k-\ell}10^{\ell-1}}\ket{0^k}=\ket{0^{k-i}10^{i-1}}\ket{0^{k-\ell}10^{\ell-1}}\ket{0^{k-(\ell-i)}10^{(\ell-i)-1}},\forall \ell\in[k]_0, \forall i\in [\ell]_0,
    \end{equation}
    can be realized by a standard quantum circuit of depth $O\left(\log(k)+\frac{k^2}{N+k}\right)$ using $N\ge 0 $ ancillary qubits with all-to-all qubit connectivity.
\end{lemma}
\begin{proof}
    Define $\ket{0^k 1 0^{-1}}:=\ket{0^k}$.
    We label the first three $k$ qubits as $S:=\{s_k,s_{k-1},\ldots,s_1\}$, $T:=\{t_k,t_{k-1},\ldots,t_1\}$ and $W:=\{w_k,w_{k-1},\ldots,w_1\}$. 
    {The $N$ ancillary qubits are divided into $q:=\lfloor N/k\rfloor$ parts of size $k$, which are defined as $S_j:=\{s_{k,j},s_{k-1,j},\ldots,s_{1,j}\}$, $T_j:=\{t_{k,j},t_{k-1,j},\ldots,t_{1,j}\}$ and $W_j:=\{w_{k,j},w_{k-1,j},\ldots,w_{1,j}\}$ for any $j\in[\lfloor q/3\rfloor]$.} 
     Eq. \eqref{eq:unitary_minus} can be realized by the following circuit acting on qubit sets $S,T,W$
    \begin{align}\label{eq:C2-original}
        & \underbrace{U_{add}\prod_{j=1}^k{\sf Tof}^{ s_k,t_j}_{w_j}(01)U^\dagger_{add}}_{C_1}\cdot \underbrace{\prod_{r,j:1\le r<j\le k}{\sf Tof}^{s_{r},t_{j}}_{w_{j-r}}(11)}_{C_2},
    \end{align}
    where $U_{add}=U_{add}(S-\{s_k\},s_k)$ act on qubit set $S$.
    The circuit $C_1$ transforms basis state $\ket{0^k}_S\ket{0^{k-\ell}10^{\ell-1}}_T\ket{0^k}_W$ to $\ket{0^k}_S\ket{0^{k-\ell}10^{\ell-1}}_T\ket{0^{k-\ell}10^{\ell-1}}_W$ for any $\ell\in [k]_0$ and leaves other basis states in Eq. \eqref{eq:unitary_minus} unchanged. The circuit $C_2$ transform basis state $\ket{0^{k-i}10^{i-1}}_S\ket{0^{k-\ell}10^{\ell-1}}_T\ket{0^k}_W$ to $\ket{0^{k-i}10^{i-1}}_S\ket{0^{k-\ell}10^{\ell-1}}_T\ket{0^{k-(\ell-i)}10^{(\ell-i)-1}}_W$ for any $\ell\in [k]_0$, $i\in[\ell]$ and leaves other basis states in Eq. \eqref{eq:unitary_minus} unchanged. According to Lemma \ref{lem:CNOT_sum} and the definition of circuit $C_1$, $C_1$ have depth $O(\log(k))+O(k)=O(k)$. Furthermore, changing the order of Toffoli gates leaves circuit $C_2$ unchanged. Therefore, the Toffoli gates in the circuit $C_2$ can be divided into $2k-3$ groups, $C_i^{(1)}(S,T,W)$ for $i\in [k-1]$ and $C_i^{(2)}(S,T,W)$ for $i\in [k-2]$ (see Figure \ref{fig:toffoli-arrangement} (a)):
    \begin{align}\label{eq:C2}
        &\prod_{i=1}^{\lfloor k/2\rfloor}\underbrace{\prod_{j=1}^{i}{\sf Tof}^{s_j,t_{2i+1-j}}_{w_{2(i-j)+1 }}(11)}_{C^{(1)}_i(S,T,W)}\prod_{i=\lfloor k/2\rfloor+1}^{k-1}\underbrace{\prod_{j=1}^{k-i}{\sf Tof}^{s_{i-j+1},t_{i+j}}_{w_{2j-1}}(11)}_{C^{(1)}_i(S,T,W)}
        \cdot \prod_{i=1}^{\lfloor(k-1)/2\rfloor}\underbrace{\prod_{j=1}^i {\sf Tof}^{s_j,t_{2i+2-j}}_{w_{2(i-j+1)}}(11)}_{C_i^{(2)}(S,T,W)} \cdot \prod_{i=\lfloor(k+1)/2\rfloor}^{k-2}\underbrace{\prod_{j=1}^{2\lfloor(k-1)/2\rfloor+1-i} {\sf Tof}^{s_{i-j+1},t_{i+j+1}}_{w_{2j}}(11)}_{C_i^{(2)}(S,T,W)}.
    \end{align}
    Note that Toffoli gates act on distinct qubits in each $C_i^{(1)}(S,T,W)$ and $C_i^{(2)}(S,T,W)$. Namely $C_i^{(1)}(S,T,W)$ and $C_i^{(2)}(S,T,W)$ have depth $1$, which implies circuit $C_2$ have depth $2k-3$. Therefore, $U_{minus}$ can be implemented in depth $O(k)+(2k-3)=O(k)$.

    Now we show how to reduce the circuit depth of $C_1$ and $C_2$ by using $N$ ancillary qubits. Assume that $N\ge 3k$. If $N < 3k$, we do not utilize ancillary qubits. We will show how to realize $C^{(1)}_1,C^{(1)}_2,\ldots,C^{(1)}_{k-1}$. The remaining ${\sf Tof}^{s_k,t_1}_{w_1}, {\sf Tof}^{s_k,t_2}_{w_2},\ldots, {\sf Tof}^{s_k,t_k}_{w_k}$ in circuit $C_1$ and $C^{(2)}_1,C^{(2)}_2,\ldots,C^{(2)}_{k-2}$ can be implemented in the same way.
   \begin{itemize}
       \item Step 1: We make $\lfloor q/3\rfloor$ copies of qubit sets $S$, $T$ on $S_\tau$ and $T_\tau$ by a circuit of depth $O(\log(q))$ based on Lemma \ref{lem:copy}, i.e., for any $i\in [\ell]_0$ and $\ell\in [k]_0$,
   \begin{align*}
       &\ket{0^{k-i}10^{i-1}}_S\ket{0^{k-\ell}10^{\ell-1}}_T\ket{0^k}_W\bigotimes_{\tau=1}^{\lfloor q/3\rfloor}\ket{0^k}_{S_\tau}\ket{0^k}_{T_\tau}\ket{0^k}_{W_\tau}\\
      \xrightarrow{U_{copy}} &\ket{0^{k-i}10^{i-1}}_S\ket{0^{k-\ell}10^{\ell-1}}_T\ket{0^k}_W\bigotimes_{\tau=1}^{\lfloor q/3\rfloor}\ket{0^{k-i}10^{i-1}}_{S_\tau}\ket{0^{k-\ell}10^{\ell-1}}_{T_\tau}\ket{0^k}_{W_\tau}.
   \end{align*}

   \item Step 2: For each $C_i^{(1)}(S,W,T)$, define a corresponding circuit $C_i^{(1)}(S_\tau,W_\tau,T_\tau)$ acting on qubits of $S_\tau,T_\tau,W_\tau$. If there is a Toffoli gate ${\sf Tof}^{s_a,t_b}_{w_{c}}$ in $C_i^{(1)}(S,W,T)$, then there is a Toffoli gate ${\sf Tof}^{s_{a,\tau},t_{b,\tau}}_{w_{c,\tau}}$ in $C_i^{(1)}(S_\tau,W_\tau,T_\tau)$.
   Let $d:=\left\lceil\frac{k-1}{\lfloor q/3\rfloor}\right\rceil$. To implement a $C_{j+(\tau-1)d}^{(1)}(S,W,T)$, we implement $C_{j+(\tau-1)d}^{(1)}(S_\tau,W_\tau,T_\tau)$ on qubit sets $S_\tau, T_\tau, W_\tau$ for all $j\in [d]$, which have circuit depth $d$.
   
   \item Step 3: If we add states in qubits $w_{i,1},w_{i,2},\ldots,w_{i,\lfloor q/3\rfloor}$ to qubit $w_i$ for any $i\in [k]$, then the state of $w_i$ is the same as the state which is obtained by applying $C_1^{(1)}(S,T,W),\ldots,C_{(k-1)}^{(1)}(S,T,W)$ on qubit sets $S,T,W$. The above procedure can be implemented in depth $O(\log(q))$ by Lemma \ref{lem:CNOT_sum}.

   \item Step 4: Restore all qubits in $S_\tau, T_\tau, W_\tau$ for any $\tau\in[\lfloor q/3\rfloor]$ by the inverse circuits of step 2 and 1. The total depth is $O(\log(q))+d=O(\log(q)+d)$.
   \end{itemize} 
   The total depth to implement $C_{1}^{(1)},\ldots,C_{k-1}^{(1)}$ is $O(\log(q))+d+O(\log(d))+O(\log(q)+d)=O(\log(N/k)+k^2/N)$. The $1$-depth circuits ${\sf Tof}^{s_k,t_1}_{w_1}, {\sf Tof}^{s_k,t_2}_{w_2},\ldots, {\sf Tof}^{s_k,t_k}_{w_k}$ and $C^{(2)}_1,C^{(2)}_2,\ldots,C^{(2)}_{k-2}$ can be realized in the same way of depth $O(\log(N/k)+k^2/N)$.

   In summary, $U_{minus}$ can be implemented in depth $O\left(\log(k)+\frac{k^2}{N+k}\right)$ using $N\ge 0$ ancillary qubits.
\end{proof}

\begin{lemma}\label{lem:unitary_plus}
A unitary transformation $U_{plus}$ satisfying
    \begin{equation}\label{eq:unitary_plus}
        U_{plus}\ket{0^{k-i}10^{i-1}}\ket{0^{k-(\ell-i)}10^{(\ell-i)-1}}\ket{0^k}=\ket{0^{k-i}10^{i-1}}\ket{0^{k-(\ell-i)}10^{(\ell-i)-1}}\ket{0^{k-\ell}10^{\ell-1}},\forall \ell\in[k]_0, i\in [\ell]_0,
    \end{equation}
    can be realized by a standard quantum circuit of depth $O\left(\log(k)+\frac{k^2}{N+k}\right)$ using $N\ge 0 $ ancillary qubits with all-to-all qubit connectivity.
\end{lemma}
\begin{proof}
      Let $\ket{0^k 1 0^{-1}}:=\ket{0^k}$.
    We label the first three $k$ qubits as $S:=\{s_k,s_{k-1},\ldots,s_1\}$, $T:=\{t_k,t_{k-1},\ldots,t_1\}$ and $W:=\{w_k,w_{k-1},\ldots,w_1\}$. 
     Eq. \eqref{eq:unitary_plus} can be realized by the following circuit acting on qubit sets $S,T,W$
    \begin{align}\label{eq:C2_plus-original}
        & \underbrace{U_{add}\prod_{j=1}^k{\sf Tof}^{ s_k,t_j}_{w_j}(01)U^\dagger_{add}}_{C_1}\cdot \underbrace{\prod_{r,j:1\le r<j\le k}{\sf Tof}^{s_{r},t_{j-r}}_{w_{j}}(11)}_{C_2},
    \end{align}
    where $U_{add}=U_{add}(S-\{s_k\},s_k)$ act on qubit set $S$.
    The circuit $C_1$ transforms basis state $\ket{0^k}_S\ket{0^{k-\ell}10^{\ell-1}}_T\ket{0^k}_W$ to $\ket{0^k}_S\ket{0^{k-\ell}10^{\ell-1}}_T\ket{0^{k-\ell}10^{\ell-1}}_W$ for any $\ell\in [k]_0$ and leaves other basis states in Eq. \eqref{eq:unitary_plus} unchanged. The circuit $C_2$ transform basis state $\ket{0^{k-i}10^{i-1}}_S\ket{0^{k-\ell}10^{\ell-1}}_T\ket{0^k}_W$ to $\ket{0^{k-i}10^{i-1}}_S\ket{0^{k-\ell}10^{\ell-1}}_T\ket{0^{k-(\ell-i)}10^{(\ell-i)-1}}_W$ for any $\ell\in [k]_0$, $i\in[\ell]$ and leaves other basis states in Eq. \eqref{eq:unitary_plus} unchanged. The Toffoli gates in the circuit $C_2$ can be divided into $2k-3$ groups, $C_i^{(1)}(S,T,W)$ for $i\in [k-1]$ and $C_i^{(2)}(S,T,W)$ for $i\in [k-2]$ (see Figure \ref{fig:toffoli-arrangement}(b)):
        \begin{align}\label{eq:C2_plus}
        &\prod_{i=1}^{\lfloor k/2\rfloor}\underbrace{\prod_{j=1}^{i}{\sf Tof}^{s_j,t_{2(i-j)+1 }}_{w_{2i+1-j}}(11)}_{C^{(1)}_i(S,T,W)}\prod_{i=\lfloor k/2\rfloor+1}^{k-1}\underbrace{\prod_{j=1}^{k-i}{\sf Tof}^{s_{i-j+1},t_{2j-1}}_{w_{i+j}}(11)}_{C^{(1)}_i(S,T,W)}
        \cdot \prod_{i=1}^{\lfloor(k-1)/2\rfloor}\underbrace{\prod_{j=1}^i {\sf Tof}^{s_j,t_{2(i-j+1)}}_{w_{2i+2-j}}(11)}_{C_i^{(2)}(S,T,W)} \cdot \prod_{i=\lfloor(k+1)/2\rfloor}^{k-2}\underbrace{\prod_{j=1}^{2\lfloor(k-1)/2\rfloor+1-i} {\sf Tof}^{s_{i-j+1},t_{2j}}_{w_{i+j+1}}(11)}_{C_i^{(2)}(S,T,W)}.
    \end{align}
    The above circuit has the same form of Eq. \eqref{eq:C2}.
    therefore, by the same discussion of Lemma \ref{lem:unitary_minus}, $U_{plus}$ can be implemented in depth $O\left(\log(k)+\frac{k^2}{N+k}\right)$ using $N\ge 0$ ancillary qubits.
\end{proof}

With the above tools, we can reduce the circuit depth of the divide unitary.
\begin{lemma}\label{lem:WDB_ancilla}
    The divide unitary $\D^{n,m}_k(S_1,S_2)$ defined as in Eq.\eqref{eq:wdb} can be implemented by a standard quantum circuit of depth $O\left(\log(k)+\frac{k^2}{k+N}\right)$ using $N$ ($\ge 0$) ancillary qubits. 
\end{lemma}
\begin{proof}
    Let $\ket{0^{k-\ell}10^{-1}}:=\ket{0^k}$.
    If the number of ancillary qubits $N< 2k$, we implement $\D^{n,m}_k$ by a circuit of depth $O(k)$ according to Lemma \ref{lem:WDB_path}. If the number of ancillary qubits $N\ge 2k$, we implement $\D^{n,m}_k$ as follows. For any $\ell\in[k]_0$,
    \begin{align}
     &\ket{0^k}\ket{0^{k-\ell} 1^\ell}\ket{0^N} \nonumber \\
     \to  &\ket{0^k}\ket{0^{k-\lceil \log(k+1)\rceil} (\ell)_2} \ket{0^N} \label{eq:1} & \text{(by Lemma \ref{lem:binary})}\\
     \to & \frac{1}{\sqrt{\binom{n}{\ell}}}\sum_{i=0}^{\ell} \sqrt{\binom{m}{i}\binom{n-m}{\ell-i}}\ket{0^{k-\lceil \log(k+1)\rceil}(i)_2}\ket{0^{k-\lceil \log(k+1)\rceil} (\ell)_2}\ket{0^N} & \text{(by Lemma \ref{lem:multi-QSP})} \label{eq:2}\\
     \to & \frac{1}{\sqrt{\binom{n}{\ell}}}\sum_{i=0}^{\ell} \sqrt{\binom{m}{i}\binom{n-m}{\ell-i}}\ket{0^{k-i}10^{i-1}}\ket{0^{k-\ell}10^{\ell-1}}\ket{0^N} & \text{(by the inverse of Eq. \eqref{eq:U_ob})} \label{eq:3}\\
     \to & \frac{1}{\sqrt{\binom{n}{\ell}}}\sum_{i=0}^{\ell} \sqrt{\binom{m}{i}\binom{n-m}{\ell-i}}\ket{0^{k-i}10^{i-1}}\ket{0^{k-\ell}10^{\ell-1}}\ket{0^{k-(\ell-i)}10^{(\ell-i)-1}}\ket{0^{N-k}} & \text{(by Lemma \ref{lem:unitary_minus})} \label{eq:4}\\
     \to & \frac{1}{\sqrt{\binom{n}{\ell}}}\sum_{i=0}^{\ell} \sqrt{\binom{m}{i}\binom{n-m}{\ell-i}}\ket{0^{k-i}10^{i-1}}\ket{0^k}\ket{0^{k-(\ell-i)}10^{(\ell-i)-1}} \ket{0^{N-k}} & \text{(by Lemma \ref{lem:unitary_plus})} \label{eq:5}\\
     \to &\frac{1}{\sqrt{\binom{n}{\ell}}}\sum_{i=0}^{\ell} \sqrt{\binom{m}{i}\binom{n-m}{\ell-i}}\ket{0^{k-i}1^i}\ket{0^{k+i-\ell}1^{\ell-i}}\ket{0^k}\ket{0^{N-k}} & \text{(by the inverse of Eq. \eqref{eq:U_uo})} \label{eq:6}\\
     = & \D^{n,m}_k\ket{0^k}\ket{0^{k-\ell} 1^\ell}\ket{0^N} & \text{(by Eq. \eqref{eq:wdb})}
\end{align}
Based on Lemma \ref{lem:binary}, the circuit depth of Eq. \eqref{eq:1} is $O\big(\log(k)+\frac{k^2}{(N+k)\log(N+k)}\big)+O\big(\log(k)+\frac{k^2}{N+k}\big)=O\big(\log(k)+\frac{k^2}{N+k}\big)$ by using both Eqs. \eqref{eq:U_uo} and \eqref{eq:U_ob}. Eq. \eqref{eq:2} is a $(\lceil\log(k+1)\rceil,\lceil\log(k+1)\rceil)$-CQSP, which can be realized by a circuit of depth $O\left(\log(k)+\frac{k^2}{N+\log(k)}\right)$ using $N$ ancillary qubits based on  Lemma \ref{lem:multi-QSP}. Eq. \eqref{eq:3} can be implemented by applying the inverse circuits of Eq. \eqref{eq:U_uo} of depth $O\big(\log(k)+\frac{k^2}{N+k}\big)$ by using $N$ ancillary qubits. Eq. \eqref{eq:4} can be realized in depth $O\big(\log(k)+\frac{k^2}{N}\big)$ by using $N-k$ ancillary qubits. Eq. \eqref{eq:5} can be implemented by a inverse circuit of $U_{plus}$ by using $N-k$ ancillary qubits in Lemma \ref{lem:unitary_plus}, which has depth $O\big(\log(k)+\frac{k^2}{N}\big)$. To implement Eq. \eqref{eq:6}, first we swap the second and the third $k$ qubits by a swap circuit of depth $1$; second we apply the inverse circuit of Eq. \eqref{eq:U_uo} using $N-k$ ancillary qubits, which has depth $O\big(\log(k)+\frac{k^2}{N\log(N)}\big)$ based on Lemma \ref{lem:binary}.
Hence, if there are $N\ge 2k$ ancillary qubits, $\D^{n,m}_k$ can be implemented by a circuit of depth
$2\cdot O\big(\log(k)+\frac{k^2}{N+k}\big)+O\big(\log(k)+\frac{k^2}{N+\log(k)}\big)+2\cdot O\big(\log(k)+\frac{k^2}{N}\big)+O\big(\log(k)+\frac{k^2}{N\log(N)}\big)= O\big(\log(k)+\frac{k^2}{N+k}\big)$.
 In summary, $\D^{n,m}_k$ can be implemented by a circuit of depth $O\big(\log(k)+\frac{k^2}{k+N}\big)$ using $N\ge0$ ancillary qubits.
\end{proof}

\begin{figure}[htbp]
        \centering
        \subcaptionbox{Circuits $C_i^{(1)}(S,T,W)$ and $C_i^{(2)}(S,T,W)$ in Eq. \eqref{eq:C2}.}{
        \includegraphics[width = .45\textwidth]{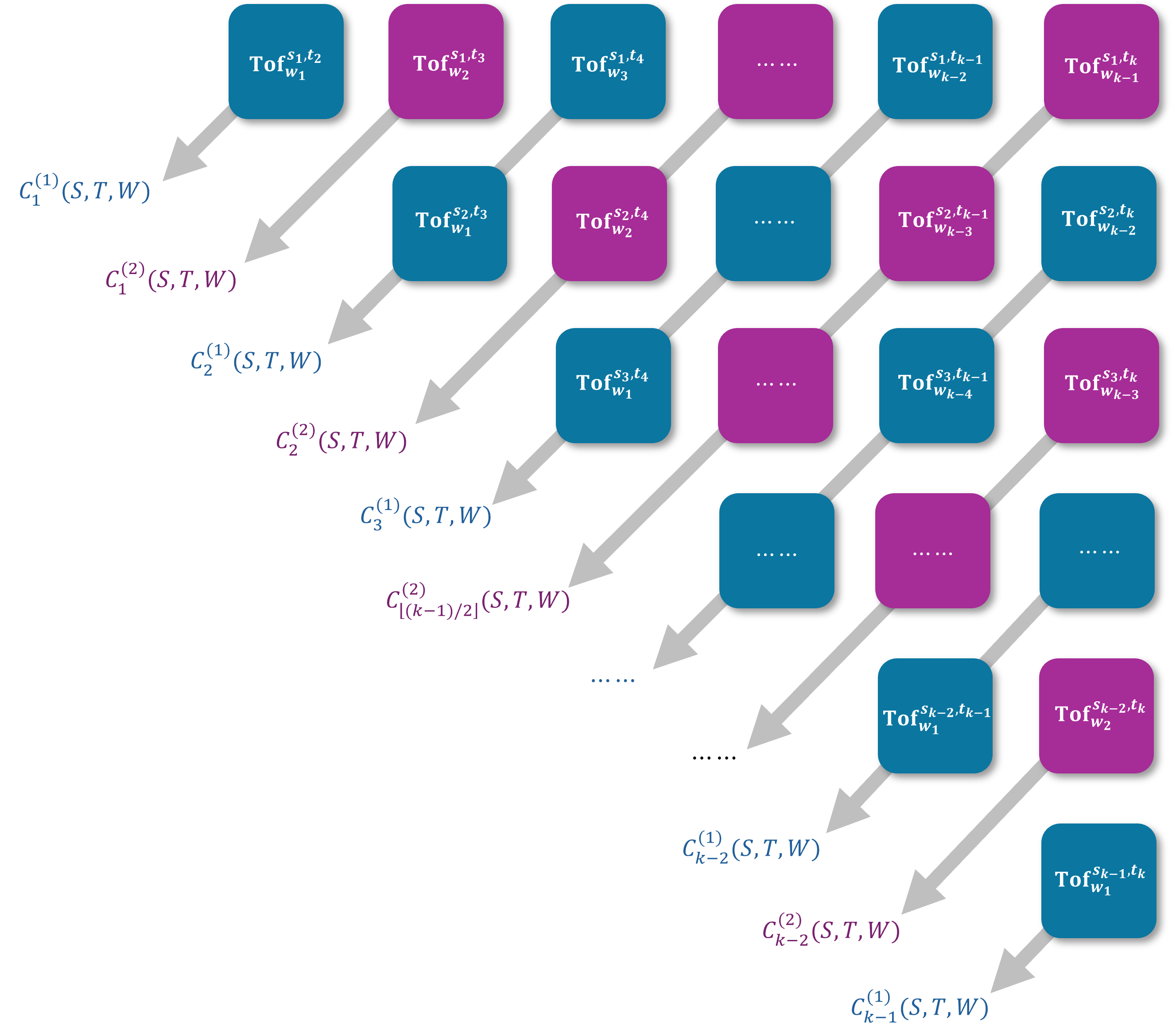}
    }
        \subcaptionbox{Circuits $C_i^{(1)}(S,T,W)$ and $C_i^{(2)}(S,T,W)$ in Eq. \eqref{eq:C2_plus}.}{
        \centering
        \includegraphics[width = .45\textwidth]{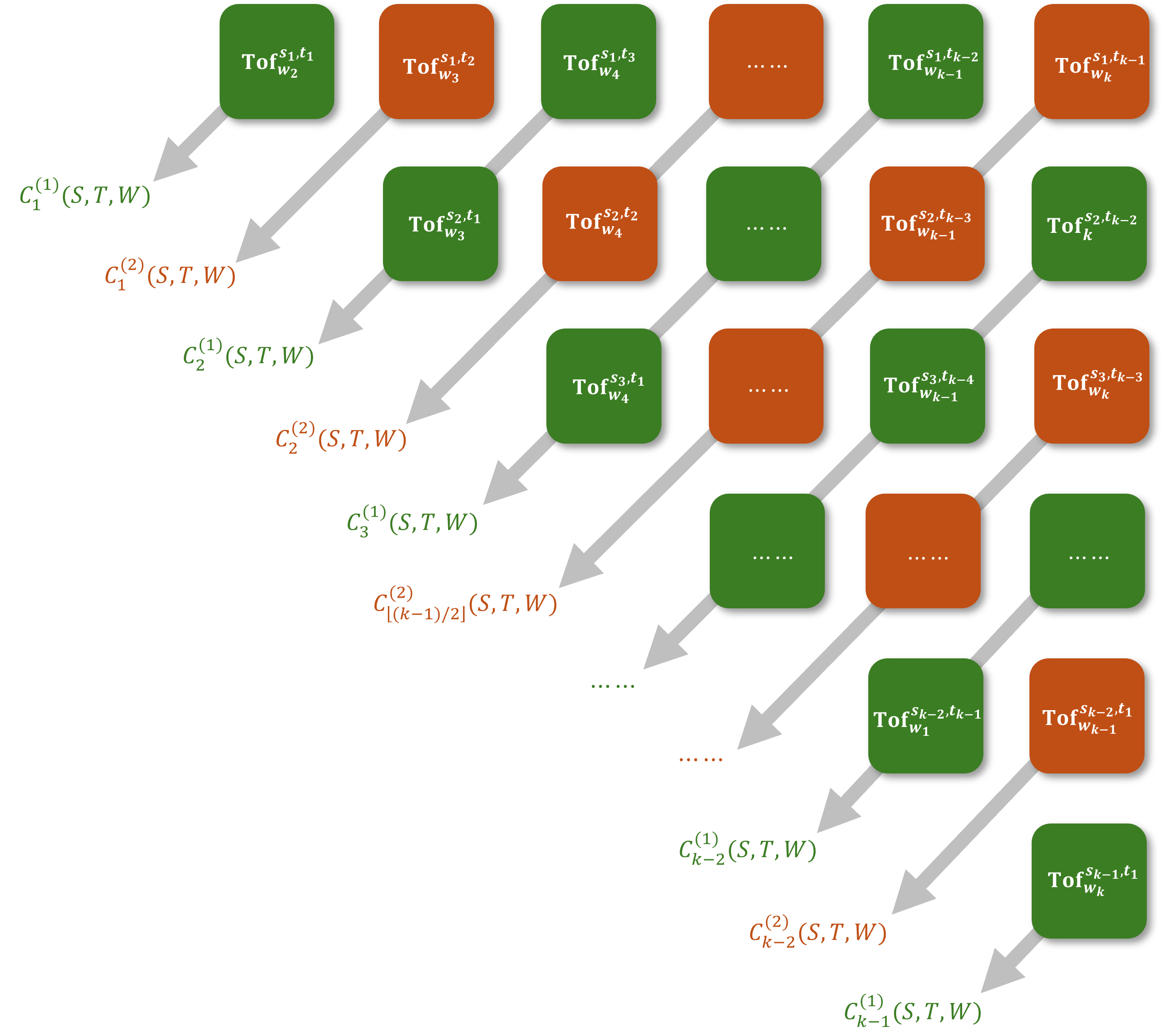}
    }
        \caption{The blocks in (a) and (b) are all Toffoli gates in Eqs.~\eqref{eq:C2-original} and \eqref{eq:C2_plus-original}, respectively. The Toffoli gates on the same arrow (the Toffoli gates in $C_i^{(1)}$ or $C_i^{(2)}$) have distinct control and target qubits.}
        \label{fig:toffoli-arrangement}
    \end{figure}

Lemma \ref{lem:WDB_ancilla} can then be used to construct efficient circuit preparing the Dicke state.
\begin{theorem}\label{thm:Dicke_unitary_noconstraint_improve}
    The Dicke state $\ket{D^n_k}$ can be prepared by a standard quantum circuit of depth $O\big(\log(k)\log(n/k)+k\big)$ with all-to-all qubit connectivity.
\end{theorem}
\begin{proof}
Any $\ket{D^n_k}$ can be prepared by applying $\U^n_k$ on state $\ket{0^{n-k}1^k}$.
If $k=\Omega(n)$, $\ket{D^n_k}$ can be realized by a circuit of depth $O(k)$ according to Lemma \ref{lem:Dicke_unitary_path}. If $k=o(n)$, $\U^n_k$ will be implemented as follows.
As previously discussed, a Dicke state unitary can be implemented by a circuit consisting of $d$ layers of $2k$-qubit divide unitaries and one layer of $\ell$-qubit Dicke state unitaries, where $d$ is at most $\lfloor\log(n/k)\rfloor$ and $\ell=O(k)$. There are $2^{j-1}$ divide unitaries $\D^{m_1,m_2}_k$ where $m_1\ge m_2$, $m_1=O(n/2^{j-1})$ and $m_2=O(n/2^{j})$ in the $j$-layer for $j\in[d]$ and $2^{d}$ $O(k)$-qubit Dicke state unitaries in the $(d+1)$-th layers. Since each divide unitary acts on $2k$ qubits, $n-2^{j-1}\cdot 2k=n-2^j k$ are idle in the $j$-th layer, which can be utilized as the ancillary qubits. In the $j$-th layer, each divide unitary is located $N_j=\lfloor(n-2^j k)/2^{j-1}\rfloor=\lfloor n/2^{j-1}\rfloor-2k$ ancillary qubits.
Therefore, the total depth for $\U^n_k$ acting on $\ket{0^{n-k}1^k}$ is 
   \begin{align*}
       &\sum_{j=1}^d\underbrace{ O\left(\log(k)+\frac{k^2}{k+N_j}\right)}_{\text{the depth of the $j$-th layer}}+\underbrace{O(\ell)}_{\text{the depth of the $(d+1)$-th layer}}\\
       = &O(\log(n/k)\log(k))+\sum_{j=1}^dO\left(\frac{k^2}{n/2^{j-1}}\right)+O(n/2^{\log(n/k)-2})\\
     = & O(\log(n/k)\log(k)+k),
   \end{align*}
   where in the first line, the first and second term are obtained by Lemmas \ref{lem:WDB_ancilla} and \ref{lem:WDB_path}.
\end{proof}

% The following corollary can be induced by Theorem \ref{thm:Dicke_unitary_noconstraint_improve}. 
% \begin{corollary}\label{coro:Dicke_noconstraint_improve}
%     The Dicke state $\ket{D^n_k}$ can be prepared by a quantum circuit of depth $O(\log(k)\log(n/k)+k)$ without qubit connectivity constraints.
% \end{corollary}
% \begin{proof}
%     We apply $X^{\otimes k}$ on the first $k$ qubits of initial state $\ket{0^n}$ and apply a Dicke state unitary $\U^n_k$. Based on Theorem \ref{thm:Dicke_unitary_noconstraint_improve}, the total depth is $O(\log(k)\log(n/k)+k)$.
% \end{proof}

%{\noindent\bf Remark.} If $k=O(1)$ or $k=\Omega(\log(n)\log(\log(n)))$, the circuit depth for Dicke state preparation is $O(k+\log(n/k))$.

\subsection{Dicke state preparation under grid qubit connectivity constraint}
\label{sec:grid}
{We now present an improved construction for implementing  Dicke state preparation under grid $\Grid_n^{n_1,n_2}$ constraint, achieving better depth scaling than prior work \cite{bartschi2022short} while handling all parameter regimes.} 
{When $n_2/n_1\le k \le n/2$, Ref. \cite{bartschi2022short} achieves a circuit depth of $O(\sqrt{nk})$ for $\U^n_{k}$ using the same framework as in Eq.~\eqref{eq:framework} albeit with an unbalanced decomposition: They first decomposed Dicke state unitary as $\U^n_k=(\U^{\sqrt{nk}}_k\otimes \U^{n-\sqrt{nk}}_k) \D^{n, \sqrt{nk}}_k$ and then recursively implemented $\U^{\sqrt{nk}}_k$ and $\U^{n-\sqrt{nk}}_k$. Note that the divide unitary needs to be implemented on $2k$ adjacent qubits (Lemma \ref{lem:WDB_path}), %In addition, unlike the balanced decomposition framework discussed at the beginning of Section \ref{sec:circuit}, 
{and this unbalanced decomposition is easy to implement under the connectivity constraint. However, the unbalanced recursion leads to a large overall depth. We improve upon \cite{bartschi2022short} by a balanced recursion which can employ better parallelization. The price is that balanced decomposition requires to move 1s to the middle of the current row (or column) in the grid, which brings extra overhead. But we will show that the balanced approach yields greater parallelization benefits than the positioning overhead, resulting in an overall reduction of depth. This is formalized in the next theorem, which not only improves the result in \cite{bartschi2022short}, but also optimally handles the case of $k < n_2/n_1$, which was not studied in \cite{bartschi2022short}.} 
%To reduce the depth, we adapt the circuit framework from Section \ref{sec:circuit} while addressing the non-consecutive qubit issue for divide unitaries. \syz{Add one sentence to explain the intuition how the construction solves the issue with lower depth.} First, we partition the $\Grid_n^{n_1,n_2}$ grid into subgrids $\Grid^{\sqrt{{n_1k}/{n_2}}, \sqrt{{n_2k}/{n_1}}}_k$ if $n_2/n_1\le k \le n/2$ and $\Grid_k^{1,k}$ if $1\le k \le n_2/n_1$.  The $\Grid_n^{n_1,n_2}$ can also be partitioned into two identical grids on the left and right. Second, we apply $\D^{n,n/2}_k$ on two consecutive subgrids at the grid's upper-left corner, then permute the second subgrid's state to the right subgrid's upper-left region. Finally, we implement $\U^{n/2}_k$ on both left and right grids using the same recursive methodology.
}
%The $(n_1\times n_2)$-size grid is partitioned into multiple smaller grids, each containing $k$ qubits. The Dicke state is prepared using the circuit framework introduced at the beginning of Section \ref{sec:circuit}, which is distinct from the one described in Ref. \cite{bartschi2022short}.
%The division operators in the circuit framework of \cite{bartschi2022short} exhibit insufficient parallelism, which leads to an excessive circuit depth.
%Our circuit framework effectively circumvents this issue. However, the $2k$ qubits associated with the divide operators are non-consecutive within this framework. To address this, in each layer consisting of divide operators, the operators are first implemented on two adjacent small grids. Subsequently, one grid is swapped into its designated position as defined by the circuit framework. This  ensures that divide operators in subsequent layers can be executed in parallel. \syz{This paragraph is actually quite good and can serve as the summary of the result. We may simply drop the previous paragraph.}

% We can implement any $n$-qubit circuit under $\Grid_n^{n_1,n_2}$ constraints by introducing depth overhead $O(n_2)$. The following theorem shows an improvement in this circuit depth. 
%The following shows the implementation of the Dicke state unitary under the grid constraint.
\begin{theorem}\label{thm:Dicke_unitary}
    The Dicke state unitary $\U^n_k$ can be implemented by a standard quantum circuit of depth $O(k\log(n/k)+n_2)$ if $k\ge n_2/n_1$, and of depth $O(n_2)$ if $k< n_2/n_1$ under $\Grid_n^{n_1, n_2}$ constraint.
\end{theorem}
\begin{proof}
    We consider two cases: $k\ge n_2/n_1$ and $k< n_2/n_1$. \\
    { \bf Case 1: $k\ge n_2/n_1$.} First, we show a partition of qubits on grid $\Grid_n^{n_1,n_2}$, see Fig. \ref{fig:grid_partition}. The grid $\Grid_n^{n_1,n_2}$ is partitioned into $n/k$ small grids of size $\sqrt{\frac{n_1k}{n_2}} \times \sqrt{\frac{n_2k}{n_1}}=k$, denoted by $\Grid_k^{\sqrt{\frac{n_1k}{n_2}},\sqrt{\frac{n_2k}{n_1}}}$. Let $r:=\log(\sqrt{n/k})$. In each column and row, there are $2^r~(\sqrt{n/k})$ small grids respectively. The qubit set of the smallest grid in the $i$-th row and $j$-th column is denoted by $S_{i,j}$ for any $i,j\in[2^r]$. Note that $\sqrt{\frac{n_1k}{n_2}},~ \sqrt{\frac{n_2k}{n_1}},~\sqrt{n/k},~n/k$ and $r$ are usually not integers. In practice, we can choose their ceiling values as the actual values, which do not change the order of the final circuit depth. Hence, we assume here that they are all integers for simplicity in this proof.

    Second, we show how to implement the Dicke state unitary under the $\Grid_n^{n_1,n_2}$ constraint. Recall that a divide unitary $\D^{n,m}_k$ can be implemented by a quantum circuit of depth $O(k)$ on $2k$ adjacent qubits constrained by $\Path_{2k}$ if $n,m\ge k$ according to Lemma \ref{lem:Dicke_unitary_path}. Let $P(S,S')$ be a permutation unitary that exchanges the state in $S$ and $S'$ of size $k$. We partition the grid $\Grid_n^{n_1,n_2}$ into left and right grids. The qubit sets of left and right grids are $S_L:=\bigcup_{i\in [2^r],j\in [2^{r-1}]}S_{i,j}$ and $S_R:=\bigcup_{i\in [2^r],j\in [2^r]-[2^{r-1}]}S_{i,j}$.
    Now we show the circuit implementation of a Dicke state unitary $\U^n_k(S_L\cup S_R)$. For any $\ell\in[k]_0$,
     \begin{align*}
          & ~\bigotimes_{(i,j)\in[2^{r}]^2\atop (i,j)\notin\{(1,1),(1.2)\}}\ket{0^k}_{S_{i,j}}\ket{0^{k}}_{S_{1,1}}\ket{0^{k-\ell}1^\ell}_{S_{1,2}}\\
       \xrightarrow{\D^{n,n/2}_k(S_{1,1},~S_{1,2})}&~ \bigotimes_{(i,j)\in[2^{r}]^2\atop (i,j)\notin\{(1,1),(1.2)\}}\ket{0^k}_{S_{i,j}}\frac{1}{\sqrt{\binom{n}{\ell}}}\sum_{i=0}^{\ell} \sqrt{\binom{n/2}{i}\binom{n/2}{\ell-i}}\ket{0^{k-i}1^i}_{S_{1,1}}\ket{0^{k+i-\ell}1^{\ell-i}}_{S_{1,2}} & \text{(by Eq. \eqref{eq:wdb})}\\
        \xrightarrow{P(S_{1,2}, ~S_{1,2^{r-1}+1})}&~ \bigotimes_{(i,j)\in[2^{r}]^2\atop (i,j)\notin\{(1,1),(1.2^{r-1}+1)\}}\ket{0^k}_{S_{i,j}}\frac{1}{\sqrt{\binom{n}{\ell}}}\sum_{i=0}^{\ell} \sqrt{\binom{n/2}{i}\binom{n/2}{\ell-i}}\ket{0^{k-i}1^i}_{S_{1,1}}\ket{0^{k+i-\ell}1^{\ell-i}}_{S_{1,2^{r-1}+1}} & \text{(by Eq. \eqref{eq:permutation})}\\
       \xrightarrow{\U^{n/2}_{k}(S_L)\otimes \U^{n/2}_{k}(S_R)} &~  \frac{1}{\sqrt{\binom{n}{\ell}}}\sum_{i=0}^{\ell} \sqrt{\binom{n/2}{i}\binom{n/2}{\ell-i}}\ket{D^{n/2}_i}_{S_L}\ket{D^{n/2}_{\ell-i}}_{S_R} & \text{(by Eq. \eqref{eq:Dicke_unitary})}\\
       =&~ \frac{1}{\sqrt{\binom{n}{\ell}}} \sum_{x: x\in \{0,1\}^n,\atop |x|=\ell}\ket{x}_{S_L\cup S_R}=\ket{D^n_\ell}_{S_L\cup S_R} & \text{(by Eq. \eqref{eq:Dicke_state})}\\
       =&~ \U^n_k(S_L\cup S_R)\ket{0^{n-\ell}1^\ell}_{S_L\cup S_R}. &\text{(by Eq. \eqref{eq:Dicke_unitary})}
     \end{align*}
In $\Grid^{n_1,n_2}_n$, sets $S_{1,1}$ and $S_{1,2}$ are two adjacent small grids, and there is a path including all qubits of them. Based on Lemma \ref{lem:WDB_path}, we first apply $\D^{n,n/2}_k(S_{1,1},S_{1,2})$ which can be implemented by a circuit of depth $O(k)$ under $\Path_{2k}$ constraint. Second, we exchange the qubit in $S_{1,2}$ and $S_{1,2^{r-1}+1}$ by a circuit of depth $O(n_2/2)$ under the grid constraint according to Lemma \ref{lem:permutation}. Note that $S_{1,1}$ and $S_{1,2^{r-1}+1}$ are located at the top left corners of two grids $S_L$ and $S_R$ respectively. Third, we can apply $\U^{n/2}_k(S_L)$ and $\U^{n/2}_k(S_R)$ simultaneously since they act on distinct grids. Similar to the discussion in the proof Lemma 1, $\U^{n/2}_k(S_L)$ and $\U^{n/2}_k(S_R)$ can be implemented recursively.
For any $i\in[r]$, we divide the grid of size $n_1\times (n_2/2^{i-1}) $ into two equal grids of size $n_1\times (n_2/2^{i}) $ along the vertical direction in the $i$-th recursive step. In the $i$-th step, first we apply a $O(k)$-depth divide unitary on the first two smallest grids in the first row of the left grid according to Lemma \ref{lem:WDB_path}. Second, we permute the second qubit set to the top left qubit set of the right grid by a circuit of depth $O(n_2/2^i)$ by Lemma \ref{lem:permutation}. After all $r$ recursive steps, we only need to simultaneously implement a sequence of $(\sqrt{nk},k)$-Dicke state unitaries acting on grids of size $n_1\times \sqrt{n_2k/n_1}$. Furthermore, these $(\sqrt{nk},k)$-Dicke state unitaries can be implemented recursively in the same way by dividing the grids into two grids, one on top and one at the bottom. Then after $r$ recursive steps,  $(\sqrt{nk},k)$-Dicke state unitaries are decomposed as some divided operators and $k$-qubit Dicke state unitaries.
Let $T(n)$ denote the circuit depth of an $n$-qubit Dicke state unitary. According to Lemma \ref{lem:Dicke_unitary_path}, $T(k)=O(k)$. Then we have 
\begin{align*}
    T(n)= &T(n/2)+O(k)+O(n_2)/2\\
        = & T(n/2^{r})+r\cdot O(k)+\sum_{i=1}^{r} O(n_2/2^i)\\
        =&T(n/2^{r})+O(k\log(n/k))+O(n_2)\\
        = &T(n/2^{r+1})+O(k)+O(n_1/2)+O(k\log(n/k))+O(n_2)\\
        = & T(n/2^{2r})+r\cdot O(k)+\sum_{i=1}^r O(n_1/2^i)+O(k\log(n/k))+O(n_2)\\
        = &T(k)+O(n_1)+O(n_2)+O(k\log(n/k))\\
        =&O(k\log(n/k)+n_2).    
\end{align*}

      {\noindent\bf Case 2: $k< n_2/n_1$.} We partition the grid $\Grid_n^{n_1,n_2}$ into $n/k$ smallest grids $\Grid_k^{1,k}$. For simplicity, assume that $n/k$ and $n_2/k$ are integers. In each row and column, there are $n_2/k$ and $n_1$ smallest grids. The small grid in the $i$-th row and $j$-th column is denoted by $S_{i,j}$ for any $i\in[n_1]$ and $j\in [n_2/k]$. We define $S_j:=\bigcup_{i=1}^{n_1}S_{i,j}$ consisting of all smallest grids in the $j$-th column. Now we show the circuit implementation of $\U^n_k(\bigcup_{j=1}^{n_2/k}S_j)$. For any $\ell\in[k]_0$,
    \begin{align}
       & \ket{0^k}_{S_{1,1}}\ket{0^{k-\ell}1^\ell}_{S_{1,2}} \bigotimes_{(i,j)\in[n_1]\times [n_2/k]\atop (i,j)\neq (1,1),(1,2)}\ket{0^k}_{S_{i,j}}\nonumber\\
      \xrightarrow{\D^{n,n_1k}_k(S_{1,1},S_{1,2})} & \frac{1}{\sqrt{\binom{n}{\ell}}}\sum_{\tau=0}^{\ell}\sqrt{\binom{n_1k}{\tau}\binom{n-n_1k}{\ell-\tau}}\ket{0^{k-\tau}1^\tau}_{S_{1,1}}\ket{0^{n-n_1k+\tau-\ell}1^{\ell-\tau}}_{S_{1,2}}  \bigotimes_{(i,j)\in[n_1]\times [n_2/k]\atop (i,j)\neq (1,1),(1,2)}\ket{0^k}_{S_{i,j}} & \text{(by Eq. \eqref{eq:wdb})}\\
      = & \frac{1}{\sqrt{\binom{n}{\ell}}}\sum_{\tau=0}^{\ell}\sqrt{\binom{n_1k}{\tau}\binom{n-n_1k}{\ell-\tau}}\ket{0^{n_1k-\tau}1^\tau}_{S_{1}}\ket{0^{n-n_1k+\tau-\ell}1^{\ell-\tau}}_{\bigcup_{j=2}^{n_2/k}S_{j}}\\
      \xrightarrow{\U^{n_1k}_k(S_1)\otimes \U^{n-n_1k}_k(\bigcup_{j=2}^{n_2/k}S_{j})} & \frac{1}{\sqrt{\binom{n}{\ell}}}\sum_{\tau=0}^{\ell}\sqrt{\binom{n_1k}{\tau}\binom{n-n_1k}{\ell-\tau}}\ket{D^{n_1k}_\tau}_{S_{1}}\ket{D^{n-n_1k}_{\ell-\tau}}_{\bigcup_{j=2}^{n_2/k}S_{j}} &\text{(by Eq. \eqref{eq:Dicke_unitary})}\\
      =& \ket{D^n_\ell}_{\bigcup_{j=1}^{n_2/k}S_j} &\text{(by Eq. \eqref{eq:Dicke_state})}
    \end{align}
    As discussed above, we apply a $\D^{n,n_1k}_k(S_{1,1},S_{1,2})$ on qubit sets $S_{1,1}$ and $S_{1,2}$ and then apply Dicke state unitaries $\U^{n_1k}_k(S_1)$ and $\U^{n-n_1k}_k(\bigcup_{j=2}^{n_2/k}S_{j})$ on qubit sets $S_1$ and $\bigcup_{j=2}^{n_2/k}S_{j}$ respectively. Furthermore, the Dicke state unitary $\U^{n-n_1k}_k(\bigcup_{j=2}^{n_2/k}S_{j})$ can be implemented in the same way. 
    First, we apply a $\D^{n-n_1k,n_1k}(S_{1,2},S_{1,3})$ on qubit sets $S_{1,2}$ and $S_{1,3}$ and then apply Dicke state unitaries $\U^{n_1k}_k(S_2)$ and $\U^{n-2n_1k}_k(\bigcup_{j=3}^{n_2/k}S_{j})$ simultaneously, and so on. {Let $T(n)$ denote the circuit depth for $\U^{n}_k$. According to Lemmas \ref{lem:Dicke_unitary_path} and \ref{lem:WDB_path}, the circuit depth of $\D^{n,n_1k}_k(S_{1,1},S_{1,2})$ and $\U^{n_1k}_k$ are $O(k)$ and $O(n_1k)$ under the path constraints. Then we have }
    \begin{align*}
        T(n) = O(k)+T(n-n_1k)=2O(k)+T(n-2n_1k)=j\cdot O(k)+T(n-jn_1k)=(n_2/k-1)\cdot O(k) + T(n_1k)=O(n_2),
    \end{align*}
    where $T(n_1k)=O(n_1k)\le O(n_2)$ based on Lemma \ref{lem:Dicke_unitary_path}.
\end{proof}

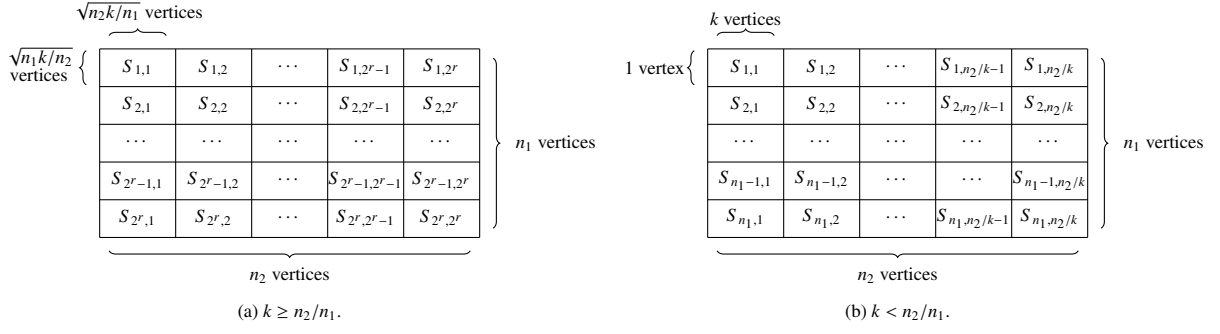
\begin{figure}
    \centering
\begin{tikzpicture}
     \centering
     \draw (0,0) -- (5,0) (0,0.5) -- (5,0.5) (0,1) -- (5,1) (0,1.5) -- (5,1.5) (0,2) -- (5,2) (0,2.5) -- (5,2.5);
     \draw (0,0) -- (0,2.5) (1,0) -- (1,2.5) (2,0) -- (2,2.5) (3,0) -- (3,2.5) (4,0) -- (4,2.5) (5,0) -- (5,2.5);

     \draw (0.5,0.25) node{\scriptsize $S_{2^r,1}$} (1.5,0.25) node{\scriptsize $S_{2^r,2}$} (2.5,0.25) node{\scriptsize $\cdots$} (3.5,0.25) node{\scriptsize $S_{2^r,2^{r}-1}$} (4.5,0.25) node{\scriptsize $S_{2^r,2^{r}}$} ;

     \draw (0.5,0.75) node{\scriptsize $S_{2^r-1,1}$} (1.5,0.75) node{\scriptsize $S_{2^r-1,2}$} (2.5,0.75) node{\scriptsize $\cdots$} (3.5,0.75) node{\scriptsize $S_{2^r-1,2^r-1}$} (4.5,0.75) node{\scriptsize $S_{2^r-1,2^r}$};

     \draw (0.5,1.25) node{\scriptsize $\cdots$} (1.5,1.25) node{\scriptsize $\cdots$} (2.5,1.25) node{\scriptsize $\cdots$} (3.5,1.25) node{\scriptsize $\cdots$} (4.5,1.25) node{\scriptsize $\cdots$};

     \draw (0.5,1.75) node{\scriptsize $S_{2,1}$} (1.5,1.75) node{\scriptsize $S_{2,2}$} (2.5,1.75) node{\scriptsize $\cdots$} (3.5,1.75) node{\scriptsize $S_{2,2^{r}-1}$} (4.5,1.75) node{\scriptsize $S_{2,2^{r}}$};

     \draw (0.5,2.25) node{\scriptsize $S_{1,1}$} (1.5,2.25) node{\scriptsize $S_{1,2}$} (2.5,2.25) node{\scriptsize $\cdots$} (3.5,2.25) node{\scriptsize $S_{1,2^{r}-1}$} (4.5,2.25) node{\scriptsize $S_{1,2^{r}}$};

    \node (a) at (5,0) {};
     \node (b) at (5,2.5) {};
     \draw[decorate,decoration={brace,raise=5pt}] (b) -- (a);

     \node (c) at (0,0) {};
     \node (d) at (5,0) {};
     \draw[decorate,decoration={brace,raise=5pt}] (d) -- (c);

      \node (e) at (0,1.9) {};
      \node (f) at (0,2.6) {};
      \draw[decorate,decoration={brace,raise=5pt}] (e) -- (f);

      \node (g) at (0,2.5) {};
      \node (h) at (1,2.5) {};
      \draw[decorate,decoration={brace,raise=5pt}] (g) -- (h);

     \draw (2.5,-0.5) node{\scriptsize $n_2$~vertices} (6,1.25) node{\scriptsize $n_1$~vertices}(-0.8,2.4) node{\scriptsize $\sqrt{n_1k/n_2}$ } (-0.8,2.15) node{\scriptsize  vertices}   (0.5,3) node{\scriptsize $\sqrt{n_2k/n_1}$ vertices} (2.5,-1) node{\scriptsize (a) $k\ge n_2/n_1$.};
% \end{tikzpicture}

% \begin{tikzpicture}
%     \centering
     \draw (8,0) -- (13,0) (8,0.5) -- (13,0.5) (8,1) -- (13,1) (8,1.5) -- (13,1.5) (8,2) -- (13,2) (8,2.5) -- (13,2.5);
     \draw (8,0) -- (8,2.5) (9,0) -- (9,2.5) (10,0) -- (10,2.5) (11,0) -- (11,2.5) (12,0) -- (12,2.5) (13,0) -- (13,2.5);

     \draw (8.5,0.25) node{\scriptsize $S_{n_1,1}$} (9.5,0.25) node{\scriptsize $S_{n_1,2}$} (10.5,0.25) node{\scriptsize $\cdots$} (11.5,0.25) node{\scriptsize $S_{n_1,n_2/k-1}$} (12.5,0.25) node{\scriptsize $S_{n_1,n_2/k}$} ;

     \draw (8.5,0.75) node{\scriptsize $S_{n_1-1,1}$} (9.5,0.75) node{\scriptsize $S_{n_1-1,2}$} (10.5,0.75) node{\scriptsize $\cdots$} (11.5,0.75) node{\scriptsize $\cdots$} (12.5,0.75) node{\scriptsize $S_{n_1-1,n_2/k}$};

     \draw (8.5,1.25) node{\scriptsize $\cdots$} (9.5,1.25) node{\scriptsize $\cdots$} (10.5,1.25) node{\scriptsize $\cdots$} (11.5,1.25) node{\scriptsize $\cdots$} (12.5,1.25) node{\scriptsize $\cdots$};

     \draw (8.5,1.75) node{\scriptsize $S_{2,1}$} (9.5,1.75) node{\scriptsize $S_{2,2}$} (10.5,1.75) node{\scriptsize $\cdots$} (11.5,1.75) node{\scriptsize $S_{2,n_2/k-1}$} (12.5,1.75) node{\scriptsize $S_{2,n_2/k}$};

     \draw (8.5,2.25) node{\scriptsize $S_{1,1}$} (9.5,2.25) node{\scriptsize $S_{1,2}$} (10.5,2.25) node{\scriptsize $\cdots$} (11.5,2.25) node{\scriptsize $S_{1,n_2/k-1}$} (12.5,2.25) node{\scriptsize $S_{1,n_2/k}$};

    \node (a) at (13,0) {};
     \node (b) at (13,2.5) {};
     \draw[decorate,decoration={brace,raise=5pt}] (b) -- (a);

     \node (c) at (8,0) {};
     \node (d) at (13,0) {};
     \draw[decorate,decoration={brace,raise=5pt}] (d) -- (c);

      \node (e) at (8,1.9) {};
      \node (f) at (8,2.6) {};
      \draw[decorate,decoration={brace,raise=5pt}] (e) -- (f);

      \node (g) at (8,2.5) {};
      \node (h) at (9,2.5) {};
      \draw[decorate,decoration={brace,raise=5pt}] (g) -- (h);

     \draw (10.5,-0.5) node{\scriptsize $n_2$~vertices} (14,1.25) node{\scriptsize $n_1$~vertices}(7.3,2.25) node{\scriptsize $1$ vertex} (8.5,2.9) node{\scriptsize $k$ vertices} (10.5,-1) node{\scriptsize (b) $k<n_2/n_1$.};
\end{tikzpicture}
    \caption{A partition of $\Grid_n^{n_1,n_2}$, each $S_{i,j}$ contains $k$ vertices. (a) If $k\ge n_2/n_1$, each $S_{i,j}$ is a grid $\Grid^{\sqrt{n_1k/n_2},\sqrt{n_2k/n_1}}_k$. (b)  If $k< n_2/n_1$, each $S_{i,j}$ is a grid $\Grid^{1,k}_k$.}
    \label{fig:grid_partition}
\end{figure}

Theorem \ref{thm:Dicke_unitary} immediately implies the following result of the Dicke state preparation.
\begin{corollary}\label{coro:Dicke_grid_improve}
     The $(n,k)$-Dicke state $\ket{D^n_k}$ can be prepared by a standard quantum circuit of depth $O(k\log(n/k)+n_2)$ if $k\ge n_2/n_1$, and of depth $O(n_2)$ if $k\le n_2/n_1$ under $\Grid_n^{n_1, n_2}$ constraint.
\end{corollary}

\subsection{Low-level symmetric states}
Since the Dicke states $\{\ket{D^n_\ell}:\forall \ell\in[n]_0\}$ form an orthonormal basis for the symmetric subspace, any symmetric quantum state can be expressed as $\sum_{\ell=0}^n \alpha_\ell \ket{D^n_\ell}$ for some coefficients $\alpha_\ell\in \mathbb{C}$ with $\sum_{\ell=0}^n |\alpha_\ell|^2=1$. The circuits constructed in Section \ref{sec:all-to-all} and \ref{sec:grid} can be used to prepare for any symmetric state composed of low-level basis. More precisely, an $n$-qubit symmetric state $\ket{\Psi_k^n}$ is at level at most $k$ if 
\begin{equation}
\ket{\Psi_k^n}=\sum_{\ell=0}^{k}\alpha_\ell \ket{D^{n}_\ell},
\end{equation}
where $\alpha_\ell\in\mathbb{C}$ for any $\ell\in [k]_0$ and $\sum_{\ell=0}^k|\alpha_\ell|^2=1$.

\begin{lemma}[\cite{bartschi2019deterministic}]\label{lem:unary-state}
    Any $k$-qubit quantum state $\sum_{\ell=0}^k \alpha_\ell \ket{0^{k-\ell}1^\ell}$ can be prepared by a quantum circuit of depth $O(k)$ under the $\Path_k$ constraint, using no ancillary qubits.
\end{lemma}
Based on Theorems \ref{thm:Dicke_unitary_noconstraint_improve}, \ref{thm:Dicke_unitary} and Lemma \ref{lem:unary-state}, the circuit depth of the low-level symmetric state is shown as follows.
\begin{corollary}
    For any $k\in [\lfloor n/2 \rfloor]$, any $n$-qubit symmetric quantum state $ \ket{\Psi_k^n}$ at level at most $k$ can be prepared in depth $O(\log(k)\log(n/k)+k)$ for all-to-all qubit connectivity; under the $\Grid^{n_1,n_2}_n$ connectivity constraint $\ket{\Psi_k^n}$ can be prepared in depth $O(k\log(n/k)+n_2)$ if $k\ge n_2/n_1$ and $O(n_2)$ if $k<n_2/n_1$.
\end{corollary}
\begin{proof}
    The state $ \ket{\Psi_k^n}$ can be prepared in two steps. First, we prepare a $k$-qubit quantum state $\ket{\phi}=\sum_{\ell=0}^k \alpha_\ell \ket{0^{k-\ell}1^\ell}$, which can be achieved by a circuit of depth $O(k)$ under the $\Path_k$  constraint based on Lemma \ref{lem:unary-state}. Second, by applying Dicke state unitary $\U^n_k$ to $\ket{0^{n-k}}\ket{\phi}$, we obtain the target state $\ket{\Psi_k^n}$. According to Theorems \ref{thm:Dicke_unitary_noconstraint_improve} and \ref{thm:Dicke_unitary}, the circuit depth to prepare $\ket{\Psi_k^n}$ is $O(\log(k)\log(n/k)+k)$ for all-to-all qubit connectivity; $O(k\log(n/k)+n_2)$ if $k\ge n_2/n_1$ and $O(n_2)$ if $k<n_2/n_1$ under the $\Grid^{n_1,n_2}_n$ connectivity constraint.    
\end{proof}

\section{Depth lower bound for Dicke state preparation}
\label{sec:lowerbound}
In this section, we show the fundamental limits on the circuit depth for Dicke state preparation under various qubit connectivity constraint. Our analysis employs light cone arguments to quantify how quantum information propagates through constrained architectures. 

First, we review the definitions of directed graphs for quantum circuits and reachable subsets as introduced in \cite{yuan2024does}.

\begin{definition}[Directed graphs for quantum circuits]\label{def:circuit-digraph_main}
Let $C$ be a quantum circuit on $n$ qubits consisting of $d$ depth-1 layers, with odd layers consisting only of single-qubit gates, even layers consisting only of CNOT gates, and any two (non-identity) single-qubit gates acting on the same qubit must be separated by at least one CNOT gate acting on that qubit.  Let $L_1,L_2,\cdots,L_d$ denote the $d$ layers of this circuit, i.e., $C=L_dL_{d-1}\cdots L_1$.  Define the directed graph $H=(V_C,E_C)$ associated with $C$ as follows. 
\begin{enumerate}
    \item Vertex set $V_C$: For each $i\in[d+1]$, define $S_{i}:=\{v_i^j:j\in[n]\}$, where $v_i^j$ is a label corresponding to the $j$-th qubit {at time step $i$}. Let $V_C:=\bigcup_{i=1}^{d+1} S_i$. 
    \item Edge set $E_C$: For all $i\in [d]$:
    \begin{enumerate}
        \item If there is a single-qubit gate acting on the $j$-th qubit in layer $L_i$ then, for all $i \le i' \le d$ there exists a directed edge $(v_{i'+1}^j,v_{i'}^j)$.
        
        \item If there is a CNOT gate acting on qubits $j_1$ and $j_2$ in layer $L_i$, then there exist $4$ directed edges $(v_{i+1}^{j_1},v_i^{j_1})$, $(v_{i+1}^{j_2}, v_i^{j_1})$, $(v_{i+1}^{j_1},v_i^{j_2})$ and $(v_{i+1}^{j_2},v_i^{j_2})$.
    \end{enumerate}
  Note that edges are directed from $S_{i+1}$ to $S_i$. 
\end{enumerate}
\end{definition}

\begin{definition}[Reachable subsets of one qubit]\label{def:reachable} Let $H=(V_C,E_C)$ be the directed graph associated with quantum circuit $C$ of depth $d$, with vertex set $V_C = \bigcup_{i=1}^{d+1} S_i$.  For each $i\in[d+1]$ define the reachable subsets $S'_{i}$ of $H$ as follows:
\begin{itemize}
    \item $S'_{d+1} = \{v^j_{d+1}\}$ for some $j\in [n]$, i.e., the subset of a vertex in $S_{d+1}$ corresponding to the the $j$-th input qubit.
    \item For $i\in[d]$, $S'_{i}\subseteq S_i$ is the subset of vertices $v^j_i$ in $S_i$ which are (i) reachable by a directed path from vertices in $S'_{d+1}$, and (ii) there is a quantum gate acting on qubit $j$ in circuit layer $L_{i}$. 
\end{itemize}
\end{definition}

Second, we show the depth lower bound of the Dicke state.

\begin{theorem}\label{thm:lb_diam}
    Any standard quantum circuit generating the $(n,k)$-Dicke state $\ket{D^n_k}$ needs depth at least
    \begin{enumerate}
        \item $\Omega(\log(n))$ with all-to-all qubit connectivity;
        
        \item $\Omega(n_2)$ under $\Grid^{n_1,n_2}_n$ constraints;
        
        \item $\Omega(n)$ under $\Path_n$ constraints.
    \end{enumerate}
\end{theorem}
\begin{proof}
     The basic idea is, for any deterministic quantum circuit, to consider the light cones of the qubits at the last layer. If the circuit depth is not large enough, then there are two qubits whose light cones do not intersect, which makes the two qubits unentangled at the end of the circuit, if the starting state is product state. But it is not hard to verify that any two qubits are entangled in the Dicke state, therefore the depth needs to be large. For different constraint graphs the light cone expands at different paces, resulting in different lower bounds. 
     
     Next, we make the argument more precise. Let $C=L_dL_{d-1}\cdots L_1$ denote a depth-$d$ circuit for preparing Dicke state $\ket{D^n_k}$. The qubits of $C$ are labeled as $\{1,2,\cdots, n\}$. Let $H=(V_C,E_C)$ be the directed graph associated with quantum circuit $C$ of depth $d$. For each $i\in[d+1]$, define $S_{i}:=\{v_i^j:j\in[n]\}$ %, where $v_i^j$ is a label corresponding to the $j$-th qubit. Let 
     and $S_i'$ %denote the reachable set (
     as in Definitions \ref{def:circuit-digraph_main} and \ref{def:reachable}.
    \begin{enumerate}
        \item {\bf Complete graph $K_n$.} By Definition \ref{def:circuit-digraph_main}, if there is a CNOT gate acting on qubits $j_1$ and $j_2$ in layer $L_i$, then there exist $4$ directed edges $(v_{i+1}^{j_1},v_i^{j_1})$, $(v_{i+1}^{j_2}, v_i^{j_1})$, $(v_{i+1}^{j_1},v_i^{j_2})$ and $(v_{i+1}^{j_2},v_i^{j_2})$. Then for a complete graph, $|S'_i|\le 2|S'_{i+1}|$ for $1\le i \le \log(n)$. Therefore, the size of the reachable set $S'_i$ of any qubit is $\abs{S'_{i}} \le 
    O(2^{d-i+2})$ if $1\le d-i+2\le \log(n)$.
    %      \begin{equation}\label{eq:complete-Si}
    % \abs{S'_{i}} \le 
    % O(2^{d-i+2}), \quad \text{if }1\le d-i+2\le \log(n)
    % \end{equation} 
    Assume that $d=o(\log(n))$. 
    Then for the directed graph $H_C$, we can find two sequences of reachable sets $P'_{d+1}\subseteq P'_{d}\subseteq\ldots\subseteq P'_{1}$ and $Q'_{d+1}\subseteq Q'_{d}\subseteq\ldots\subseteq Q'_{1}$ such that (i) $P'_{d+1}\neq Q'_{d+1}$, (ii) $P'_{1}\cap Q'_1=\emptyset$ and (iii) $p:=|P_1'|=o(n)$ and $q:=|Q'_1|=o(n)$. Without loss of generality, let $P'_{d+1}=\{v_{d+1}^1\}$ and $ Q'_{d+1}=\{v_{d+1}^n\}$. For any $i\in [d]$, let
    \[L_i=L_i(P_i')\otimes L_i(S_i-P_i'\cup Q_i')\otimes L_i(Q'_i)\]
    where $L_i(P_i')$, $L_i(S_i-P_i'\cup Q_i')$ and $L_i(Q'_i)$ consist of all quantum gates of $L_i$ acting on qubit sets $P_i'$, $S_i-P_i'\cup Q_i'$ and $Q_i'$ respectively.
    Therefore, quantum circuit $C$ for the Dicke state can be represented as
    \begin{align}
        C = & L_d L_{d-1} \cdots L_2 L_1, \nonumber\\
          = & (L_d(P_d')\otimes L_d(S_d-P_d'\cup Q_d')\otimes L_d(Q'_d)) \cdots (L_1(P_1')\otimes L_1(S_1-P_1'\cup Q_1')\otimes L_1(Q'_1))\\
          = & (\mathbb{I}_1\otimes V \otimes \mathbb{I}_1)(U_1\otimes \mathbb{I}_{n-(p+q)}\otimes U_2),\nonumber
    \end{align}
    where $\mathbb{I}_j$ denotes an identity operator acting on $j$ qubits and 
    \begin{align*}
        &U_1:=L_d(P_d')L_{d-1}(P_{d-1}')\cdots L_1(P_1')\\
        &U_2:=L_d(Q_d')L_{d-1}(Q_{d-1}')\cdots L_1(Q_1')\\
        &V:= L_d(S_d-P_d'\cup Q_d')L_{d-1}(S_{d-1}-P_{d-1}'\cup Q_{d-1}')\cdots L_1(S_1-P_1'\cup Q_1').
    \end{align*}
    For simplicity, we omit all identity operators in $U_1$, $U_2$ and $V$. Define $\ket{\phi_1}_{\{1,2,\ldots,p\}}:=U_1\ket{0^p}_{\{1,2,\ldots,p\}}$ and $\ket{\phi_2}_{\{n-q+1,\ldots,n\}}:=U_2\ket{0^q}_{\{n-q+1,\ldots,n\}}$. By Schmidt decomposition, state $\ket{\phi_1}$ and $\ket{\phi_2}$ can be decomposed as  
    \begin{align*}
        \ket{\phi_1}_{\{1,2,\ldots,p\}}&=\lambda_0\ket{\alpha_0}_{\{1\}}\ket{\beta_0}_{\{2,3,\ldots,p\}}+\lambda_1\ket{\alpha_1}_{\{1\}}\ket{\beta_1}_{\{2,3,\ldots,p\}},\\
        \ket{\phi_2}_{\{n-q+1,\ldots,n\}}&=\sigma_0\ket{\gamma_0}_{\{n-q+1,\ldots,n-1\}}\ket{\zeta_0}_{\{n\}}+\sigma_1\ket{\gamma_1}_{\{n-q+1,\ldots,n-1\}}\ket{\zeta_1}_{\{n\}},
    \end{align*}
    where $\lambda_0^2+\lambda_1^2=1$, $\sigma_0^2+\sigma_1^2=1$ and $\{\ket{\alpha_0},\ket{\alpha_1}\}$, $\{\ket{\beta_0},\ket{\beta_1}\}$, $\{\ket{\gamma_0},\ket{\gamma_1}\}$ and $\{\ket{\zeta_0},\ket{\zeta_1}\}$ are orthogonal vector sets.
    Then we have
    \begin{align*}
        C\ket{0^{n}}_{\{1,2,\ldots,n\}}&=(\mathbb{I}_1\otimes V \otimes \mathbb{I}_1)(U_1\otimes \mathbb{I}_{n-(p+q)}\otimes U_2)\ket{0^n}_{\{1,2,\cdots, n\}}\\
        &=(\mathbb{I}_1\otimes V \otimes \mathbb{I}_1)(\ket{\phi_1}_{\{1,2,\cdots,p\}}\otimes \ket{0^{n-(p+q)}}_{\{p+1,p+2,\ldots,n-q\}}\otimes \ket{\phi_2}_{\{n-q+1,\ldots,n\}})\\
        &=(\mathbb{I}_1\otimes V \otimes \mathbb{I}_1)\sum_{i,j=0}^1\lambda_i\sigma_j\ket{\alpha_i}_{\{1\}}\ket{\beta_i 0^{n-(p+q)}\gamma_j}_{\{2,3,\ldots,n-1\}}\ket{\zeta_j}_{\{n\}}\\
        &=\sum_{i,j=0}^1\lambda_i\sigma_j\ket{\alpha_i}_{\{1\}}(V\ket{\beta_i 0^{n-(p+q)}\gamma_j}_{\{2,3,\ldots,n-1\}})\ket{\zeta_j}_{\{n\}}
    \end{align*}
    Since $\{\ket{\beta_0},\ket{\beta_1}\}$ and $\{\ket{\gamma_0},\ket{\gamma_1}\}$ are orthogonal sets, after tracing out qubits $2,3,\ldots,n-1$ of $C\ket{0^n}$, we have
     \begin{align}
        &{\rm Tr}_{\{2,3,\ldots,n-1\}}(C\ket{0^{n}}\bra{0^{n}}_{\{1,2,\ldots,n\}}C^\dagger)\nonumber\\
        =&{\rm Tr}_{\{2,3,\ldots,n-1\}}(\sum_{i,j=0}^1\lambda_i\sigma_j\ket{\alpha_i}\bra{\alpha_i}_{\{1\}}(V\ket{\beta_i 0^{n-(p+q)}\gamma_j}\bra{\beta_i 0^{n-(p+q)}\gamma_j}_{\{2,3,\ldots,n-1\}}V^\dagger)\ket{\zeta_j}\bra{\zeta_j}_{\{n\}}) \nonumber\\
         =&\sum_{i,j=0}^1\lambda_i\sigma_j\ket{\alpha_i}\bra{\alpha_i}_{\{1\}}\ket{\zeta_j}\bra{\zeta_j}_{\{n\}} \nonumber\\
         =&(\sum_{i=0}^1\lambda_i\ket{\alpha_i}\bra{\alpha_i}_{\{1\}})\otimes (\sum_{j=0}^1\sigma_j\ket{\zeta_j}\bra{\zeta_j}_{\{n\}}).\label{eq:circuit_trace}
    \end{align}
    Namely, it can be represented as a tensor product of two (mixed) states. For an $(n,k)$-Dicke state $\ket{D^n_k}$, after tracing out qubits in $\{2,3,\ldots,n-1\}$, we have
    \begin{align}
       {\rm Tr}_{\{2,3,\ldots,n-1\}}( \ket{D^n_k}\bra{D^n_k}_{\{1,2,\ldots,n\}})
    = &
    \frac{1}{\binom{n}{k}}(\binom{n-2}{k}\ket{0^2}\bra{0^2}_{\{1,n\}}+2\binom{n-2}{k-1}\ket{D_1^2}\bra{D^2_1}_{\{1,n\}}+\binom{n-2}{k-2}\ket{1^2}\bra{1^2}_{\{1,n\}})
    \end{align}
    where $\binom{n-2}{k-2}=0$ if $k=1$. The matrix representation of ${\rm Tr}_{\{2,3,\ldots,n-1\}}( \ket{D^n_k}\bra{D^n_k}_{\{1,2,\ldots,n\}})$ with respect to the orthonmal basis $\{\ket{00},\ket{01},\ket{10},\ket{11}\}$ is 
    \begin{align}
     \frac{1}{\binom{n}{k}}\left[\begin{array}{cccc}
         \binom{n-2}{k} & & &\\
          &\binom{n-2}{k-1} & \binom{n-2}{k-1}&\\
          &\binom{n-2}{k-1} & \binom{n-2}{k-1}&\\
          & & & \binom{n-2}{k-2}
     \end{array}\right],\label{eq:dicke_trace}
    \end{align}
    Assume that Eq. \eqref{eq:dicke_trace} can be represented as a tensor product of two mixed states. Then for some integers $s$ and $t$, Eq. \eqref{eq:dicke_trace} can be represented as
    \begin{equation}\label{eq:tensor_representation}
       \rho_1\otimes \rho_2=\left(\sum_{i=1}^s p_i\left[\begin{array}{cc}
          a_i   & x_i \\
           x_i^*  & 1-a_i
        \end{array}\right]\right)\otimes\left(\sum_{j=1}^tq_j\left[\begin{array}{cc}
          b_j   & y_j \\
          y^*_j  & 1-b_i
        \end{array}\right]\right),
    \end{equation}
    where $p_i,q_j, a_i, b_i\in [0,1]$, $\sum_{i=1}^s p_i=1$ and $\sum_{j=1}^t q_j=1$. Since Eqs. \eqref{eq:dicke_trace} and \eqref{eq:tensor_representation} are equivalent, we have
    \begin{align*}
        & (\sum_{i=1}^s p_i a_i)(\sum_{j=1}^tq_j b_j)=\binom{n-2}{k}/\binom{n}{k}>0,\\
        & (\sum_{i=1}^s p_i a_i)(\sum_{j=1}^tq_j y^*_j)=0,\\
        &(\sum_{i=1}^s p_i x_i)(\sum_{j=1}^tq_j y^*_j)=\binom{n-2}{k-1}/\binom{n}{k}>0.
    \end{align*}
   The first two equations imply $\sum_{j=1}^tq_j y^*_j=0$, but the last equation implies $\sum_{j=1}^tq_j y^*_j>0$. Therefore, the above equations have no solution, i.e., Eq. \eqref{eq:dicke_trace} can not be represented as a tensor product of two mixed states. Hence our assumption $d=o(\log(n))$ is not valid.
   
        \item {\bf Grid $\Grid_n^{n_1,n_2}$.} Assume that $S'_{d+1}$ is the set of the upper left (lower right) vertex of the grid.
    Note that, for $\Grid_{n}^{n_1, n_2}$, $S'_{d+1} \subseteq [1]\times [1]$, $S'_{d}\subseteq [2]\times[2]$, and so on. We have the following bounds for $|S_i'|$,
     \begin{equation}\label{eq:grid-Si}
    \abs{S'_{i}} \le \begin{cases}
    O((d-i+2)^2), &\quad \text{if }d-i+2\le n_1,\\
    O(n_1(d-i+2)), &\quad \text{if } n_{1}< d-i+2\le n_{2},\\
    n_1 n_2=n,& \quad \text{if } d-i+2>n_2.
    \end{cases}
    \end{equation}
    Assume that $d=o(n_2)$. Therefore, for the directed graph $H_C$, we can find two sequences of reachable sets $P'_{d+1}\subseteq P'_{d}\subseteq\ldots\subseteq P'_{1}$ and $Q'_{d+1}\subseteq Q'_{d}\subseteq\ldots\subseteq Q'_{1}$ such that (i) $P'_{d+1}\neq Q'_{d+1}$, (ii) $P'_{1}\cap Q'_1=\emptyset$ and (iii) $p:=|P_1'|=o(n)$ and $q:=|Q'_1|=o(n)$. By the same discussion above, we can show that $d=\Omega(n_2)$.
  
        \item {\bf Path $\Path_n$.} A path $\Path_n$ is a grid $\Grid_n^{1,n}$. Therefore, the depth lower bound is $\Omega(n)$.
    \end{enumerate}
\end{proof}

{\noindent \textbf{Remark.}}    For the circuit depth of $(n,k)$-Dicke state ($k\le n/2$), combining Theorem \ref{thm:Dicke_unitary_noconstraint_improve}, Corollary \ref{coro:Dicke_grid_improve} and Theorem \ref{thm:lb_diam}, the following conclusions can be drawn:
If there are no qubit connectivity constraints, the depth of Theorem \ref{thm:Dicke_unitary_noconstraint_improve} is asymptotically optimal when $k=O(1)$.
If there are qubit connectivity constraints $\Grid_n^{n_1,n_2}$, the depth of Corollary  \ref{coro:Dicke_grid_improve} is asymptotically optimal when $k=O(n_2/\log(n_1))$.

\section{Conclusion}
\label{sec:conclusion}
In this paper, we have shown that any $(n,k)$-Dicke state ($k\le n/2$) can be prepared by a quantum circuit consisting of single-qubit and CNOT gates of depth $O(\log(k)\log(n/k)+k)$ with all-to-all qubit connectivity. Under the $\Grid_n^{n_1,n_2}$ qubit connectivity constraint $n_1\le n_2$, we construct circuits of depth $O(k\log(n/k)+n_2)$ if $k\ge n_2/n_1$ and $O(n_2)$ if $k < n_2/n_1$. Furthermore, we also presented the depth lower bounds $\Omega(\log(n))$ and $\Omega(n_2)$ with all-to-all qubit connectivity and under $\Grid_{n}^{n_1,n_2}$ constraint, respectively. A prominent open problem is to close the gap between the depth upper and lower bounds, for which we conjecture that $\Omega(k)$ is a lower bound even for all-to-all qubit connectivity. This, if true, implies that our constructions are all optimal (up to a logarithm factor) for all-to-all connectivity and under $\Grid_{n}^{n_1,n_2}$ constraint with different parameter regimes.

\bibliographystyle{alpha}
\bibliography{reference} 

\end{document}